\documentclass[11pt]{article}

\usepackage[centering]{geometry}

\usepackage{bm}                     
\usepackage{amsfonts}
\usepackage{amssymb}
\usepackage{amsmath}
\usepackage{amsthm}
\usepackage{fixmath}
\usepackage{upgreek}
\usepackage{graphicx}               
\usepackage{pstricks,pst-plot}
\usepackage{wrapfig}
\usepackage{colortbl}               
\usepackage{booktabs}

\usepackage{mathptmx}

\newtheorem{theorem}{Theorem}
\newtheorem{lemma}{Lemma}

\newtheorem{definition}{Definition}
\newtheorem{remark}{Remark}
\newtheorem{proposition}{Proposition}

\begin{document}

\begin{center}
\large \bf Recovering the shape of a point cloud in the plane \normalsize
\end{center}
\normalsize

\

\begin{center}
  Beatriz Pateiro-L\'opez and Alberto Rodr\'{\i}guez Casal\footnote{Email: alberto.rodriguez.casal@usc.es}\\ 
  Universidad de Santiago de Compostela, Spain\\
\end{center}

%
%
%




\section*{Abstract}
In this work  we deal  with the problem of support estimation under shape restrictions. The shape restriction we deal with is an extension of the notion of convexity named $\alpha$-convexity.
 Instead of assuming, as in the convex case, the existence of a separating hyperplane for each exterior point we assume the existence of a separating open ball with radius $\alpha$.
Given an $\alpha$-convex set $S$, the $\alpha$-convex hull of independent random points in $S$ is {the natural estimator of the set. If $\alpha$ is unknown the {\mbox{$r_n$-convex}}
hull of the sample can be considered.
We analyze the asymptotic properties of the {\mbox{$r_n$-convex}} hull estimator in the bidimensional case and obtain the convergence rate for the expected distance in measure between the set and the estimator. The geometrical complexity of the estimator and its dependence on $r_n$ is also obtained via the analysis of the
expected number of vertices of the {\mbox{$r_n$-convex}} hull.

\noindent {\bf{Keywords:}} Convex set; $\alpha$-convex set;  Set estimation; Distance in measure; Image analysis.

\noindent {\bf{AMS 2000 subject classifications:}} Primary 60D05; secondary 62G20.





\section{Introduction}
Let $S$ be a convex set in the plane. Starting from the classical papers by \cite{Renyi63, Renyi64}, asymptotical behavior of the convex hull of random points in $S$ has received great attention. Also, expressions for the expected area, perimeter, and number of vertices of the convex hull of a sample have been object of research. From the point of view of set estimation, the convexity assumption has been extensively considered in the literature. If we assume that the set of interest $S$ (for instance the unknown support of an absolutely continuous distribution) is convex, then the convex hull of a sample from that distribution turns out to be a good choice to recover the shape of the support. \cite{Dumbgen96} carry out the asymptotic analysis of the convex hull estimator for general dimension (in terms of the Hausdoff distance between the estimator and the set).  Computations of measures of the convex hull, such as the number of vertices or the volume become quite complicated.  In fact, most of
the known results concern the asymptotic behavior of the expected value of some interesting geometrical characteristics such as the area, perimeter or the number of vertices. Only recently the asymptotic analysis of the variance or the limit law of these quantities have been performed for general convex sets and dimension $d$ see, for instance, \cite{Reitzner05}. We refer to the surveys by \cite{Schneider88}  for the classical results on convex set estimation} and \cite{Reitzner10} for more recent results on the subject.

Convexity can be a restrictive assumption. Just as an example, it would limit the
support estimation problem to connected supports, which is clearly inadequate, for instance if several groups are presented in $S$ and we are interested in performing a cluster analysis. On
the other hand, using the convex hull as an approximation of a non-convex set leads to considerable errors in the estimation. A milder shape-restriction
 which appears in set estimation is $\alpha$-convexity, see \cite{Walther99}.This shape
 restriction assumes that a ball with radius $\alpha$ can roll freely in the complement of $S$ (see next section for a formal definition of $\alpha$-convexity).
 This work deals with the study of a natural estimator when this restriction is imposed, the $\alpha$-convex hull of the sample, that is, the smallest $\alpha$-convex set which
 contains the sample. If $\alpha$ is unknown, we may replace $\alpha$ by a sequence of parameters $r_n$ which goes to zero as $n$ tends to infinity. Some results about the asymptotic behavior of the
 $r_n$-convex hull of the sample can be found in \cite{Rodriguez07}. Here, we are concerned with the convergence rate for the
 expected distance in measure between the set and the estimator. We prove that the obtained convergence rate is sharp and cannot be improved in general. We also study the
dependence on $r_n$ of the expected number of vertices of the {\mbox{$r_n$-convex}} hull estimator. This quantity provides some information about
 the complexity of the estimator in the sense that the more vertices the estimator has, the more complex the estimator is.

The rest of the paper is organized as follows. The shape restriction and the estimator are defined in detail in Section \ref{sec2.1}. The main results are established in Section \ref{sec4.5}. All proofs are deferred to Section \ref{sec.proofs}.

\section{The estimator, the shape restriction and general tools}\label{sec2.1}

\
\noindent \it 2.1 The estimator\rm

In what follows we assume that $S$ is a (nonempty) compact set
in the bidimensional Euclidean space $\mathbb{R}^2$, equipped with the ordinary
scalar product $\langle \cdot,\cdot \rangle$
and norm $\|\cdot \|$. We also assume that a random sample
of points $X_1,\ldots,X_n$ from
a distribution $P_X$ with support $S$ is observed. The goal is to reconstruct the
set of interest $S$. Several alternatives
have been considered in the literature. For instance, under no shape restriction on $S$,
\cite{Chevalier76} and \cite{Dev80} proposed as estimator of $S$ the union of balls of radius $\varepsilon_n$ with centers in
the sample points. See \cite{Biau} for some new results about this estimator.
However, if it assumed that $S$ fulfils some
smoothness restriction then a more
efficient estimator can be provided.
Thus, under the assumption
that $S$ is convex, the convex hull of the sample is
the natural estimator.
%
 As it was mentioned in the Introduction, this paper focuses on the problem of estimating a set
under a more flexible assumption than convexity, named {\mbox{$\alpha$-convexity}}. A set $A$ is said to be {\mbox{$\alpha$-convex}} if any point that does not
belong to the set is contained in an open ball (not necessarily centered in
the point) which does not intersect the set. This recalls us the definition of convexity and the existence of a separating plane for each exterior point. In fact, a convex set is also {\mbox{$\alpha$-convex}} for any value of $\alpha$. From its
 definition it can be easily seen that a set $A$ is {\mbox{$\alpha$-convex}} if $A=C_{\alpha}(A)$ where
\begin{equation*}
C_\alpha(A)=\bigcap_{\{\mathring{B}(x,\alpha):\ \mathring{B}(x,\alpha)\cap A=\emptyset\}}{\left(\mathring{B}(x,\alpha)\right)^c}
\end{equation*}
is the {\mbox{$\alpha$-convex}} hull of the set $A$, that is, the smallest convex set which contains the set. Here
$\mathring{B}(x,r)$ denotes the open ball with center $x$ and radius $r$ and $A^{c}$ the complement of $A$. In what follows $B$ and $\mathring{B}$ stand
for $B(0,1)$ and $\mathring{B}(0,1)$, respectively. Moreover, from now on, $\overline{A}$ and $\partial A$ will denote the closure and boundary of $A$, respectively.


The {\mbox{$\alpha$-convex}} hull of a set $A$ can be also written as the closing
of the set, that is,
$$
C_{\alpha}(A)=(A\oplus r\mathring{B})\ominus r\mathring{B},
$$
where $\oplus$ and $\ominus$ denote the Minkowski addition and subtraction, respectively.
For two sets  $A, C$ the Minkowski addition is defined by $A\oplus C= \{a+c:a\in A, c\in C\}$
whereas the Minkowski subtraction is $A\ominus C= \{x:\{x\}\oplus C\subset A\}.$ For $\lambda\in\mathbb{R}$, $\lambda C=\{\lambda c: c\in C\}$.
See \cite{Serra84} for more details on these morphological operators.

Now, let us assume that $S$ is {\mbox{$\alpha$-convex}} for some $\alpha>0$. Given
a random
sample $\mathcal{X}_n=\{X_1,\ldots,X_n\}$ from $P_X$ with support $S$,  the {\mbox{$\alpha$}-convex} hull of the sample
\[
C_\alpha(\mathcal{X}_n)=(\mathcal{X}_n\oplus \alpha\mathring{B})\ominus \alpha\mathring{B}
\]
turns out to be a natural estimator for the set $S$. This estimator has the drawback of depending on the (possibly) unknown parameter $\alpha$.
This difficulty can be overcome by taking a sequence of positive numbers $\{r_n\}$ converging to zero as $n$ tends to infinity.
This ensures that $r_n\leq\alpha$ for $n$ large enough and therefore $S$  is also $r_n$-convex. For the sake of simplicity we assume that $r_n\leq\alpha$ for all $n$ and define the estimator
\begin{equation}\label{Sn.open}
S_n=C_{r_n}(\mathcal{X}_n)=(\mathcal{X}_n\oplus r_n\mathring{B})\ominus r_n\mathring{B}.
\end{equation}
Our goal is to analyze the asymptotic properties of this set estimator. Here we will consider
the distance in measure to quantify the similarity in content of $S$ and $S_n$. As measure we will use the Lebesgue measure $\mu$. Hence, the distance
between $S$ and $S_n$ is defined as
\[d_\mu(S,S_n)=\mu(S\Delta S_n)=\mu((S\setminus S_n)\cup(S_n\setminus S))=\mu(S\setminus S_n),\]
since with probability one $\mathcal{X}_n\subset S$, which implies $S_n\subset S$.

\

\noindent \it 2.2 The shape restriction\rm

The estimator (\ref{Sn.open}) was proposed in \cite{Rodriguez07}. In
that paper the convergence rate for the Hausdorff distance is
provided, under the assumption that $S$ is a smooth $\alpha$-convex
set. Apart from the $\alpha$-convexity of $S$, it is also assumed that $\overline{S^c}$
is $\alpha$-convex. Both conditions imply that $S$ belongs to Serra's
regular model. See \cite{Walther99} for an exact geometric characterization of
Serra's regular model in terms of
{\mbox{$\alpha$-convexity}} and free rolling conditions. Essentially,
a nonempty compact set $S$ belongs to Serra's regular model if,
for some $\alpha>0$,
\begin{enumerate}
\item[(R)\label{a1}]{A ball of radius $\alpha>0$ rolls freely in $S$ and in $\overline{S^{c}}$.}\label{A1}
\end{enumerate}
We say that a ball $\alpha B$ rolls freely in a closed set $A$ if for each
boundary point $a\in\partial A$ there exists some $x\in A$ such that
$a\in B(x,\alpha)\subset A$.  Note that the free rolling
condition presented here is not exactly the same as the one given in
\cite{Walther99}. In that paper it is also required that $A\ominus
\alpha B$ is path-connected in order to preserve the physical
meaning of rolling freely. We have suppressed this additional requirement in our definition of
free rolling since it will not be
necessary for our purposes. Condition (R) is enough in order to
guarantee that both {\mbox{$S$ and $\overline{S^{c}}$ are $\alpha$-convex}}. It also guarantees the existence at each point
$s\in\partial S$ of a unique outward pointing unit normal vector
$\eta(s)$ such that

\begin{equation*}
B(s-\alpha\eta(s),\alpha)\subset S \mbox{ and } B(s+\alpha\eta(s),\alpha)\subset \overline{S^c},
\end{equation*}
The proof of these geometrical facts, see Appendix A in \cite{Pateiro08},  can be thought as an alternative
proof for Remark 3 in
\cite{Walther99} referring to the validity of its Theorem 1
when the set $S$ is not assumed to be path-connected. Another implication
of Assumption (R) has to do with the concept of positive
reach of a set, not mentioned in \cite{Walther99}. \cite{Federer59}
defines the reach of a nonempty closed set $A$ in the $d$ dimensional Euclidean space,
${\textnormal{reach}}(A)$,  as the
largest $\alpha$, possibly infinity, such
that if $x\in\mathbb{R}^d$ and $d(x,A)=\inf\{\|x-y\|:\ y\in A\}<\alpha$, then
the metric projection of $x$ onto $S$ is unique. 
\cite{Federer59} provides a generalization of the Steiner's
formula for sets with positive reach. Recall that, roughly speaking,
the Steiner's formula establishes that the Lebesgue measure
of the closed {\mbox{$r$-neighbourhood}}, $B(A,r)=\{x:\ d(x,A)\leq r\}$, of
a convex set $A$ can
be expressed as a polynomial of degree at most $d$ in $r$.
Federer's result says that the same holds for sets of positive reach
and $r<{\textnormal{reach}}(A)$. It can be proved
that, under Assumption (R), the reach of both $S$ and $\overline{{S}^{c}}$ is
greater than or equal to $\alpha$. 

\

\noindent \it 2.3 Tools: Unavoidable families of sets\rm

The procedure of bounding the expected value of $d_\mu(S,S_n)$ becomes easier
if we replace the proposed estimator by
\begin{equation}\label{Sn.close}
S_n=(\mathcal{X}_n\oplus r_n{B})\ominus r_n{B}.
\end{equation}
It is important to note that, although we use the same
notation $S_n$ for both $(\mathcal{X}_n\oplus r_n{B})\ominus r_n{B}$ and $(\mathcal{X}_n\oplus r_n\mathring{B})\ominus r_n\mathring{B}$, both
estimators are not necessarily equal, see Figure \ref{fig:noiqual}.
\begin{figure}[htb]
\begin{center}
\setlength{\unitlength}{1mm}
\begin{picture}(130,45)
\psdot(1.1330,1)
\psdot(2,2.5)
\psdot(2.83,1)

\put(10,5){\small{$X_1$}}
\put(26,5){\small{$X_2$}}
\put(20,26){\small{$X_3$}}

\psdot(5.1330,1)
\psdot(6,2.5)
\psdot(6.83,1)
\psdot(5.1330,1)
\psdot(6,2.5)
\psdot(6.83,1)
\pscircle[linestyle=dashed](5.1330,1){1}
\pscircle[linestyle=dashed](6,2.5){1}
\pscircle[linestyle=dashed](6.83,1){1}
\psline(5.1330,1)(5.1330,0)
\put(49,5){\small{$r$}}

\psdot(9.1330,1)
\psdot(10,2.5)
\psdot(10.83,1)
\psdot(9.1330,1)
\psdot(10,2.5)
\psdot(10.83,1)
\pscircle(9.1330,1){1}
\pscircle(10,2.5){1}
\pscircle(10.83,1){1}

\psdot(10,1.5)
\put(99,17){\small{$c$}}
\end{picture}
\caption{\textit{For the point set $\mathcal{X}=\{X_1,X_2,X_3\}$, $(\mathcal{X}\oplus r\mathring{B})\ominus r\mathring{B}=\mathcal{X}$ and $(\mathcal{X}\oplus r{B})\ominus r{B}=\mathcal{X}\cup\{c\}$.}}
\label{fig:noiqual}
\end{center}
\end{figure}
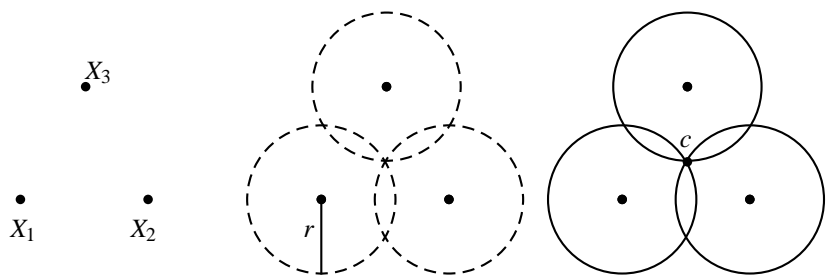
However, it is not difficult to prove that this event has probability zero, see Appendix B in \cite{Pateiro08}. Hence we can
compute $\mathbb{E}(d_\mu(S,S_n))$ by using either (\ref{Sn.open}) or (\ref{Sn.close}). Then, we can write

\begin{eqnarray}\label{esp}
\mathbb{E}(d_\mu(S,S_n))&=&\mathbb{E}(\mu(S\setminus S_n))=\int_S{P(x\notin S_n)\mu(dx)}\nonumber\\
&=&\int_S{P(\exists y\in {B}(x,r_n): {B}(y,r_n)\cap\mathcal{X}_n=\emptyset)\mu(dx)}.
\end{eqnarray}

So, the goal is to find a bound
for $P(\exists y\in {B}(x,r_n): {B}(y,r_n)\cap\mathcal{X}_n=\emptyset)$. This
bound will also allow us to obtain
a bound for the expected number of extreme points of $S_n$. As in the convex case, it is said that a sample point $X_i$ is an extreme point if  $X_i\in \partial S_n$. The number
of extreme points provide us with information about the complexity of the estimator. In
the convex case, removing the extreme points have been used in data depth for ordering multivariate data sets,
see \cite{Barnett76}. A similar idea can be used in the non-convex case. If $\mathcal{N}_n$ denotes the number of extreme points, then
$$
\mathbb{E}(\mathcal{N}_n)=nP(X_n\mbox{ is an extreme point}).
$$
It can be easily seen that $X_n$ is an extreme point of $S_n$
if and only if $X_n$ belongs to the boundary of an open ball with radius $r_n$ which does not intersect $\mathcal{X}_n$. So, conditioning on $X_n$, we
get
\begin{eqnarray}\label{extremes}
\mathbb{E}(\mathcal{N}_n)&=&  n\int_S P(\exists y\in \partial B(x,r_n): B(y,r_n)\cap\mathcal{X}_{n-1}=\emptyset )P_X(dx) \nonumber \\
&\leq& n \int_S P\left(\exists y\in {B}(x,r_n): {B}(y,r_n)\cap\mathcal{X}_{n-1}=\emptyset\right)P_X(dx),
\end{eqnarray}
where $\mathcal{X}_{n-1}=\{X_1,\ldots,X_{n-1}\}$. Hence,
if we were able to obtain an upper bound for
$P(\exists y\in {B}(x,r_n): {B}(y,r_n)\cap\mathcal{X}_n=\emptyset)$ we would get
a bound both for $\mathbb{E}(d_\mu(S,S_n))$ and $\mathbb{E}(\mathcal{N}_n)$. The idea
for bounding this probability is to make use of the concept of unavoidable family of sets, defined below.
\begin{definition}\label{inev2}
Let $x\in\mathbb{R}^2$, $r>0$ and $\mathcal{E}_{x,r}=\{{B}(y,r): y\in {B}(x,r)\}$\glossary{$\mathcal{E}_{x,r}$&$\{{B}(y,r): y\in {B}(x,r)\}$}. The family
of sets $\mathcal{U}_{x,r}$\glossary{$\mathcal{U}_{x,r}$&Unavoidable family of sets for $\mathcal{E}_{x,r}$} is said to be unavoidable
for $\mathcal{E}_{x,r}$ if, for all ${B}(y,r)\in\mathcal{E}_{x,r}$, there exists $U\in\mathcal{U}_{x,r}$ such that $U\subset{B}(y,r)$.
\end{definition}
As a consequence of Definition \ref{inev2}, if $\mathcal{U}_{x,r_n}$ is a finite unavoidable family of sets for $\mathcal{E}_{x,r_n}$, then
\begin{eqnarray}\label{inev}
P(\exists y\in {B}(x,r_n): {B}(y,r_n)\cap\mathcal{X}_n=\emptyset )&\leq & P(\exists U\in \mathcal{U}_{x,r_n}: U\cap\mathcal{X}_n=\emptyset)\\ \nonumber
&\leq& \sum_{U\in\mathcal{U}_{x,r_n}}{(1-P_X(U))^n}.
\end{eqnarray}

%
%
If we define for each $x\in S$ a family {\mbox{$\mathcal{U}_{x,r_n}$}} unavoidable and finite for $\mathcal{E}_{x,r_n}$ then, from (\ref{esp}) and (\ref{inev}), it follows that
\begin{equation}
\label{esp2}
\mathbb{E}(d_\mu(S,S_n))\leq \int_S{\sum_{U\in\mathcal{U}_{x,r_n}}(1-P_X(U))^n\mu(dx)}\leq \int_S{\sum_{U\in\mathcal{U}_{x,r_n}}\exp\left(-nP_X(U)\right)\mu(dx)},
\end{equation}
where in the last inequality we have used $(1-x)\leq \exp(-x)$, for $0\leq x\leq 1$. From (\ref{esp2}) it is apparent that the problem of finding an upper bound for $\mathbb{E}(d_\mu(S,S_n))$ (the same holds for $\mathbb{E}(\mathcal{N}_n)$) reduces to the problem
of finding a lower bound for $P_X(U)$, for all $U\in\mathcal{U}_{x,r_n}$. In view
of (\ref{esp2}) it would be desirable that, both
the lower bound and the number of elements of the
family $\mathcal{U}_{x,r_n}$, depend in the simplest possible
way on the point $x$. In order to find a lower
bound for $P_X(U)$ it is useful to assume that the probability
distribution $P_X$ is uniformly bounded on $S$, that is,
\begin{equation*}
\exists\delta>0\ \ \mbox{such that}\ \ P_X(C)\geq\delta\mu(C\cap S)
\end{equation*}
for all Borel set $C\subset \mathbb{R}^2$.
Crearly, this includes the uniform distribution on $S$.
\section{Main results}\label{sec4.5}
The main theorem of the paper provides
the convergence rate of the expected value of $d_\mu(S,S_n)$. The concept
of unavoidable family, introduced in Section \ref{sec2.1}, plays a
major role in the proof.  In Theorem \ref{nomejor} we show that the obtained
convergence rate cannot be improved.

\begin{theorem}\label{ordenR2}
Let $S$ be a nonempty compact subset of $\mathbb{R}^2$ such
that a ball of radius $\alpha>0$ rolls freely
in $S$ and in $\overline{S^{c}}$. Let $X$ be a random variable
with probability distribution $P_X$ and
support $S$. We assume that the probability
distribution $P_X$ satisfies that there exists $\delta>0$ such
that $P_X(C)\geq\delta\mu(C\cap S)$
for all Borel subset $C\subset\mathbb{R}^2$. Let $\mathcal{X}_n=\{X_1,\ldots,X_n\}$ be a random sample
from $X$ and let $\{r_n\}$ be a sequence of positive numbers
which does not depend on the sample such that $r_n\leq\alpha$. If
the sequence $\{r_n\}$ satisfies
\begin{equation}\label{conv.rn}
\lim_{n\rightarrow\infty}\frac{nr_n^2}{\log n}=\infty,
\end{equation}
then
\begin{equation}\label{cotaEE}
\mathbb{E}(d_\mu(S,S_n))=O\left(r_n^{-\frac{1}{3}}n^{-\frac{2}{3}}\right).
\end{equation}
\end{theorem}

\begin{remark}\label{casicasi}
\cite{Rodriguez07} proves that, if $S$ is under
the conditions of Theorem \ref{ordenR2} and $\{r_n\}$ is a
sequence of positive numbers satisfying (\ref{conv.rn}), then, for the bidimensional case, $d_\mu(S,S_n)=O(r_n^{-1}(\log n/n)^{2/3})$, almost surely. The
convergence rate of $\mathbb{E}(d_\mu(S,S_n))$ obtained here is, therefore, faster than
the obtained almost sure convergence rate of $d_\mu(S,S_n)$. Note that the logarithmic
term vanishes in (\ref{cotaEE}). Moreover, the penalty factor $r_n^{-1/3}$ is asymptotically smaller than $r_n^{-1}$.

\end{remark}

\begin{remark}

\label{Remark2}
The proof of Theorem \ref{ordenR2} relies on
Propositions \ref{lejos.R2} and \ref{cerca.R2} (see Section \ref{sec.proofs}) which provide
suitable unavoidable families of sets both for points
far from the boundary of $S$ and close to it. Most of the results
can be easily extended to the general $d$-dimensional case. However, some proofs are
much more involved and of less geometrical nature. The main difficulty
in analyzing the general case is in proving
Proposition \ref{cerca.R2}, see \cite{Pateiro08}.
\end{remark}

Next theorem shows
that the rate in Theorem \ref{ordenR2} cannot be improved.
\begin{theorem}\label{nomejor}
Under the conditions of Theorem \ref{ordenR2}, there exist sets $S$ for which
\[\liminf_{n\rightarrow\infty} r_n^{\frac{1}{3}}n^{\frac{2}{3}}\mathbb{E}(d_\mu(S,S_n))>0.\]
\end{theorem}


Finally we provide a bound for the expected number of extreme points. Note that the bound in (\ref{extremes}) for the number of vertices is almost the same as
the quantity which is bounded for the distance in measure, see (\ref{esp}). The main difference is that (\ref{extremes}) involves an integral
with respect to $P_X$ whereas (\ref{esp}) involves an integral with respect
to $\mu$. In order to bound integrals with respect to $P_X$ by integrals with respect to $\mu$ we assume that $P_X$ also satisfies
$$
\exists\beta>0, \mbox{ such that } P_X(C)\leq \beta \mu(C\cap S).
$$
Again the uniform distribution satisfies the above assumption.
\begin{theorem}\label{extremos}
Let us assume that the support $S$ and the sequence $r_n$ are under the conditions of Theorem \ref{ordenR2}. Let assume us
that the probability distribution which generates the sample
satisfies that there exit $\delta,\beta>0$ such that $\delta \mu(S\cap C)\leq P_X(C)\leq \beta \mu(S\cap C)$. Then,
\begin{equation*}
\mathbb{E}(\mathcal{N}_n)=O(r_n^{-\frac{1}{3}}n^{\frac{1}{3}}).
\end{equation*}
\end{theorem}

\section{Proofs}
\label{sec.proofs}

\begin{proof}[Theorem \ref{ordenR2}]
As it was mentioned in Remark \ref{Remark2}, Theorem \ref{ordenR2} relies on
Propositions \ref{lejos.R2} and \ref{cerca.R2}.
Proposition \ref{lejos.R2} gives the desired unavoidable families
for the points which are far away from the boundary of $S$. By points
which are far away from the boundary we mean
those points $x\in S$ such that $d(x,\partial S)> r_n/2$. Taking into
account Definition \ref{inev2}, it will not be difficult
to define a suitable family $\mathcal{U}_{x,r_n}$ in this case. We need
that, given $y\in B(x,r_n)$, there exists $U\in\mathcal{U}_{x,r_n}$
such that $U\subset B(y,r_n)$.
It would be also desirable that $U$ was totally
contained in $S$ and that $\mu(U)$ was of the maximum posible
order $r_n^2$. This would
ensure the best possible rate for $P_X(U)$.  Note
that if $x\in S$ and $d(x,\partial S)> r_n/2$, then
the ball $B(x,r_n/2)$ is fully contained in $S$. So,
the idea is to divide $B(x,r_n/2)$ into a finite number
of subsets. Here, we will consider a partition
of $B(x,r_n/2)$ into circular sectors. The choice of circular sectors
rests upon two main reasons. First, the measure
of a circular sector of $B(x,r_n/2)$ is of order $r_n^2$. Second, if
the central angle of the defined sectors is not too large, then the
resulting family $\mathcal{U}_{x,r_n}$ is unavoidable.

Before the statement of Proposition \ref{lejos.R2}, we
give the precise definition of the circular sectors
and introduce some basic notation that will be useful later. Thus,
let $\mathbb{S}_2=\{u\in\mathbb{R}^2:\|u\|=1\}$ denote
the unit circle in $\mathbb{R}^2$ and $e_2=(0,1)\in\mathbb{R}^2$.  Let $\varphi_{u,v}$ be
the angle between the (nonzero) vectors $u$ and $v$. It is understood
that $\varphi_{u,v}\in\left[0,\pi\right]$ and $\varphi_{u,v}=\varphi_{v,u}$
For $u\in\mathbb{S}_2$ and $\theta\in [0,\pi/2]$, we define the cone
$C_{u}^\theta=\{x\in\mathbb{R}^2:\left\langle x,u\right\rangle\geq\|x\|\cos\theta\}$
and the circular sector $C_{u,r}^\theta=C_{u}^\theta\cap B(0,r).$ Note
that $C_{u,r}^\theta$ is the circular sector with central
angle $2\theta$ enclosed by the radii $v_1=r\mathcal{R}_\theta(u)$ and $v_2=r\mathcal{R}_\theta^{-1}(u)$, where $\mathcal{R}_\theta:\mathbb{R}^2\longrightarrow\mathbb{R}^2$\glossary{$\mathcal{R}_\theta$, $\mathcal{R}$&Counter-clockwise rotation of angle $\theta$, $\mathcal{R}_{\pi/6}$} denotes the counter-clockwise rotation of angle $\theta$, whose associated matrix with respect to the canonical basis is
\[\left(\begin{array}{cc}\cos\theta&-\sin\theta\\
\sin\theta&\cos\theta\end{array}\right).\]
In Figure \ref{fig.Cu.2} we show an example of $C_{u,r}^\theta$.

\begin{figure}[htb]
\begin{center}
\setlength{\unitlength}{1mm}
\begin{picture}(50,50)
\pswedge[fillstyle=solid,fillcolor=lightgray](2.5,2.5){1.5}{15}{75}
\psline[linewidth=.5pt]{<->}(0,2.5)(5,2.5)
\psline[linewidth=.5pt]{<->}(2.5,0)(2.5,5)
\psline[linewidth=1pt]{->}(2.5,2.5)(3.94,2.88)
\psline[linewidth=1pt]{->}(2.5,2.5)(2.88,3.94)
\put(29,41){{{\scriptsize{$v_1$}}}}
\put(41,29){{{\scriptsize{$v_2$}}}}

\psline[linewidth=1pt]{->}(2.5,2.5)(3.2,3.2)
\pscircle[linewidth=.5pt](2.5,2.5){1}
\pscircle[linewidth=1pt](2.5,2.5){1.5}
\put(27,29){{{\scriptsize{$u$}}}}
\put(38,35){{{\scriptsize{$C_{u,r}^\theta$}}}}
\put(20,18){{{\scriptsize{$\mathbb{S}_2$}}}}
\psarc[linewidth=.5pt]{<->}(2.5,2.5){1.1}{15}{45}
\put(35.5,30.5){{{\scriptsize{$\theta$}}}}
\put(35,10){{{\scriptsize{${B}(0,r)$}}}}

\end{picture}
\caption{\textit{Circular sector $C_{u,r}^\theta$.}}
\label{fig.Cu.2}
\end{center}
\end{figure}
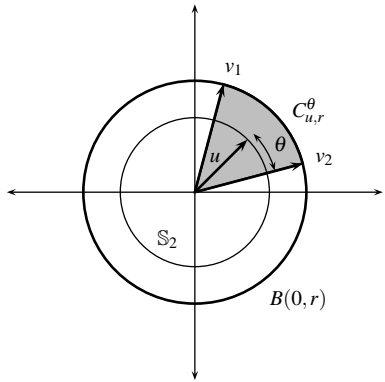

\begin{proposition}\label{lejos.R2}
Under the conditions of Theorem \ref{ordenR2}, {for all} $x\in S$
such that $d(x,\partial S)> r_n/2$, there exists a
finite family $\mathcal{U}_{x,r_n}$ with $m_1=6$ elements,
unavoidable for $\mathcal{E}_{x,r_n}$ and that satisfies
\[P_X(U)\geq L_1r_n^2,\ \ U\in \mathcal{U}_{x,r_n},\]
where the constant $L_1>0$ is independent of $x$.
\end{proposition}
\begin{remark}
This proposition can be easily generalized for dimension $d$. The main
difference is that $m_1$ is in general unknown since it depends on the number
of cones we need to cover the unit ball.
\end{remark}

\begin{proof}
First consider the family
$\mathcal{U}_{0,r_n}=\{C_{u,r_n/2}^{\pi/6},\ u\in\mathcal{W}\}$,
where $\mathcal{W}\subset\mathbb{R}^2$ denotes a set of unit vectors that divides the unit circle into six circular sectors with central angle $\pi/3$. Figure \ref{fig:W0} shows one possible choice of $\mathcal{W}$ and the corresponding family $\mathcal{U}_{0,r_n}$.
\begin{figure}
\begin{center}
\setlength{\unitlength}{1mm}
\begin{picture}(100,60)
\psline[linewidth=.5pt]{<->}(0,2.5)(5,2.5)
\psline[linewidth=.5pt]{<->}(2.5,0)(2.5,5.5)
\psline[linewidth=1pt]{->}(2.5,2.5)(3.5,4.23)
\put(30,45){{{\scriptsize{$u_3=\left(\frac{1}{2},\frac{\sqrt{3}}{2}\right)$}}}}
\psline[linewidth=1pt]{->}(2.5,2.5)(4.5,2.5)
\put(38,27.5){{{\scriptsize{$u_2=(1,0)$}}}}
\psline[linewidth=1pt]{->}(2.5,2.5)(3.5,0.77)
\put(30,4){{{\scriptsize{$u_1=\left(\frac{1}{2},-\frac{\sqrt{3}}{2}\right)$}}}}
\psline[linewidth=1pt]{->}(2.5,2.5)(1.5,4.23)
\put(3,45){{{\scriptsize{$u_4=\left(-\frac{1}{2},\frac{\sqrt{3}}{2}\right)$}}}}
\psline[linewidth=1pt]{->}(2.5,2.5)(0.5,2.5)
\put(2,27.5){{{\scriptsize{$u_5=(-1,0)$}}}}
\psline[linewidth=1pt]{->}(2.5,2.5)(1.5,0.77)
\put(1,4){{{\scriptsize{$u_6=\left(-\frac{1}{2},-\frac{\sqrt{3}}{2}\right)$}}}}

\pscustom[linewidth=.5pt,hatchwidth=.2pt,fillstyle=hlines]{
\psline(8.5,2.5)(9.53,3.1)
\psarc(8.5,2.5){1.2}{30}{90}
\psline(8.5,3.6)(8.5,2.5)
}
\pscustom[linewidth=.5pt,hatchwidth=.2pt,fillstyle=vlines]{
\psline(8.5,2.5)(8.5,3.6)
\psarc(8.5,2.5){1.2}{90}{150}
\psline(7.47,3.1)(8.5,2.5)
}
\pscustom[linewidth=.5pt,hatchwidth=.2pt,fillstyle=hlines]{
\psline(8.5,2.5)(7.47,3.1)
\psarc(8.5,2.5){1.2}{150}{210}
\psline(7.47,1.9)(8.5,2.5)
}
\pscustom[linewidth=.5pt,hatchwidth=.2pt,fillstyle=vlines]{
\psline(8.5,2.5)(7.47,1.9)
\psarc(8.5,2.5){1.2}{210}{270}
\psline(8.5,1.3)(8.5,2.5)
}
\pscustom[linewidth=.5pt,hatchwidth=.2pt,fillstyle=hlines]{
\psline(8.5,2.5)(8.5,1.3)
\psarc(8.5,2.5){1.2}{270}{330}
\psline(9.53,1.9)(8.5,2.5)
}
\pscustom[linewidth=.5pt,hatchwidth=.2pt,fillstyle=vlines]{
\psline(8.5,2.5)(9.53,1.9)
\psarc(8.5,2.5){1.2}{330}{390}
\psline(9.53,3.1)(8.5,2.5)
}

\psline[linewidth=.5pt]{<->}(6,2.5)(11,2.5)
\psline[linewidth=.5pt]{<->}(8.5,0)(8.5,5.5)
\pscircle[linewidth=.5pt](8.5,2.5){1.2}
\psline[linewidth=.5pt]{->}(8.5,2.5)(9.5,4.23)
\psline[linewidth=.5pt]{->}(8.5,2.5)(10.5,2.5)
\psline[linewidth=.5pt]{->}(8.5,2.5)(7.5,4.23)
\psline[linewidth=.5pt]{->}(8.5,2.5)(6.5,2.5)
\psline[linewidth=.5pt]{->}(8.5,2.5)(7.5,0.77)

\put(95,45){{{\scriptsize{$u_3$}}}}
\put(73,45){{{\scriptsize{$u_4$}}}}
\put(73,4){{{\scriptsize{$u_6$}}}}
\put(102,27){{{\scriptsize{$u_2$}}}}
\put(65,27){{{\scriptsize{$u_5$}}}}

\psline[linewidth=.5pt]{->}(8.5,2.5)(9.5,0.78)
\put(95,4){{{\scriptsize{$u_1$}}}}
\put(23.5,57){{{\scriptsize{(a)}}}}
\put(97,33){{{\scriptsize{$B(0,r_n/2)$}}}}
\put(83.5,57){{{\scriptsize{(b)}}}}
\end{picture}
\caption{\textit{(a) The set $\mathcal{W}=\{u_i,\ i=1,\ldots,6\}$ divides the unit circle into six circular sectors with central angle $\pi/3$. (b) Family $\mathcal{U}_{0,r_n}=\{C_{u,r_n/2}^{\pi/6}, u\in\mathcal{W}\}$.}}
\label{fig:W0}
\end{center}
\end{figure}
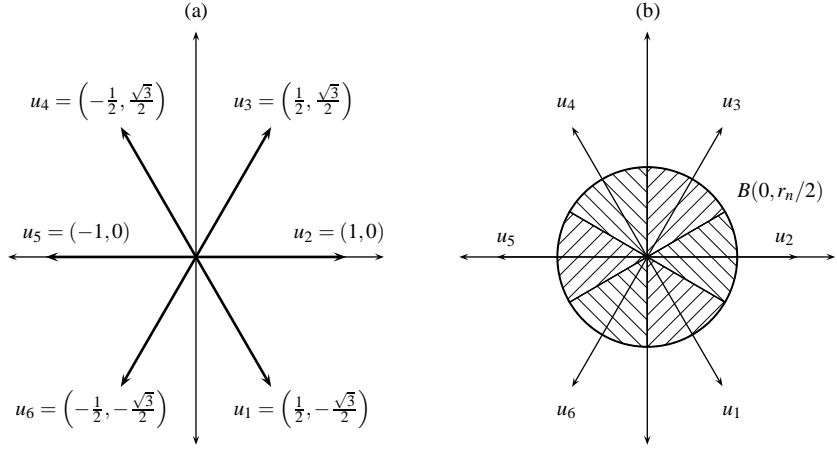
To simplify notation somewhat, we
abbreviate  $C_{u}^{\pi/6}$ and $C_{u,r_n}^{\pi/6}$ to $C_{u}$ and $C_{u,r_n}$, respectively. Note that the definition of $\mathcal{W}$ implies that
\begin{equation*}
B(0,r_n)=\bigcup_{u\in \mathcal{W}}C_{u,r_n}.
\end{equation*}
The fact that $\mathcal{U}_{0,r_n}$ is unavoidable
for $\mathcal{E}_{0,r_n}$ easily follows from Lemma \ref{cont.R2}, stated
below. To see this, note that for ${B}(y,r_n)\in\mathcal{E}_{0,r_n}$, there
exists $u\in\mathcal{W}$ such that {\mbox{$y\in C_{u,r_n}$}}. Now, by
Lemma \ref{cont.R2}, $C_{u,r_n}\subset B(y,r_n)$ and
therefore $C_{u,r_n/2}\subset B(y,r_n)$. This completes
the proof that $\mathcal{U}_{0,r_n}$ is unavoidable. Thus,
it remains to prove Lemma \ref{cont.R2}. First
we establish, without proof, Lemma \ref{estaenC.R2} which characterizes
the points in $C_u^\theta$ and simplifies the proof of Lemma \ref{cont.R2}.
\begin{lemma}\label{estaenC.R2}
Let $x\neq 0$. Then
$$x\in C_u^\theta\Leftrightarrow\varphi_{x,u}\leq\theta.$$
\end{lemma}

We are now ready to state and prove Lemma \ref{cont.R2}. This lemma reveals that
the partition of $B(0,r_n)$ into circular sectors with central angle $\pi/3$ is indeed a sensible choice, since it guarantees that $\mathcal{U}_{0,r_n}$ is unavoidable.
\begin{lemma}\label{cont.R2}
For all $u\in\mathbb{S}_2$ and $r>0$,
\[C_{u,r}\subset\bigcap_{y\in C_{u,r}}B(y,r).\]
\end{lemma}
\begin{proof}
Let $z\in C_{u,r}$. We need to show that, for all $y\in C_{u,r}$, $\left\|z-y\right\|\leq r$. Assume, without loss of generality,  that $z$ and $y$ are both non zero vectors since the result is trivial otherwise. We have that
\begin{eqnarray*}
\|z-y\|^2=\|z\|^2+\|y\|^2-2\|z\|\|y\|\cos\varphi_{z,y}.
\end{eqnarray*}
By the triangle inequality for angles and
Lemma \ref{estaenC.R2} we have $\varphi_{z,y}\leq \varphi_{z,u}+\varphi_{u,y}\leq \frac{\pi}{3}$. Hence,
\[\|z-y\|^2\leq\|z\|^2+\|y\|^2-\|z\|\|y\|\leq\max(\|z\|^2,\|y\|^2)\leq r^2.\]\qed 
\end{proof}

Once we have proved
that $\mathcal{U}_{0,r_n}$ is unavoidable
for $\mathcal{E}_{0,r_n}$ consider, for
each $x\in S$ such that $d(x,\partial S)> r_n/2$, the family
$\mathcal{U}_{x,r_n}=\{x\}\oplus \mathcal{U}_{0,r_n}=\{\{x\}\oplus C_{u,r_n/2},\ u\in\mathcal{W}\}$.
The family $\mathcal{U}_{x,r_n}$, obtained by translating the
family $\mathcal{U}_{0,r_n}$ by the vector $x$, is unavoidable
for $\mathcal{E}_{x,r_n}$, as we state in Lemma \ref{traslada2}. We skip the proof since it is straightforward.
\begin{lemma}\label{traslada2}
Let $\mathcal{U}_{0,r}$ be an unavoidable family for $\mathcal{E}_{0,r}$. Then $\mathcal{U}_{x,r}=\{x\}\oplus\mathcal{U}_{0,r}=\{\{x\}\oplus U,\ U\in \mathcal{U}_{0,r}\}$ is unavoidable for $\mathcal{E}_{x,r}$.
\end{lemma}

To complete the proof of Proposition \ref{lejos.R2} it remains to give
a lower bound for the probability of the sets of the unavoidable
family we have just defined. For each $u\in\mathcal{W}$ we
have that
$$P_X\left(\{x\}\oplus C_{u,r_n/2}\right)\geq\delta\mu\left(\{x\}\oplus C_{u,r_n/2}\cap S\right)=\delta\mu\left(\{x\}\oplus C_{u,r_n/2}\right)=\delta\mu\left(C_{u,r_n/2}\right).$$
This follows simply
because $\{x\}\oplus C_{u,r_n/2}\subset B(x,r_n/2)\subset S$
since $d(x,\partial S)> r_n/2$ and the Lebesgue measure is
invariant under translations, see Figure \ref{fig:cerca}. Therefore,
\[P_X(U)\geq \delta\frac{1}{6}\pi\left(\frac{r_n}{2}\right)^2=L_1r_n^2,\ \ U\in \mathcal{U}_{x,r_n},\]
for $L_1=\delta\pi/24>0$ and the proof of  Proposition \ref{lejos.R2} is complete. \qed

\end{proof}

\begin{figure}
\begin{center}
\setlength{\unitlength}{1mm}
\begin{picture}(115,50)
\put(70,45){{{\scriptsize{$S$}}}}
\pscustom[linewidth=1.5pt,fillstyle=solid,fillcolor=gray]{
\pscurve(8,2.5)(9,1.5)(11,3)(9.5,4)(7.5,4.5)(7.5,3.5)(8,2.5)
}
\pscircle[linewidth=.5pt](8.95,3.55){.42}
\pscircle[linewidth=.5pt](8.95,3.55){.84}
\put(87,35){{{\tiny{$x$}}}}
\pscustom[linewidth=.5pt,fillstyle=solid,fillcolor=black]{
\pswedge(8.95,3.55){.41}{30}{90}
}
\put(82.9,30){{{\tiny{$B(x,r_n/2)$}}}}
\put(56,30){{{{$\longrightarrow$}}}}
\put(56,33){{{\scriptsize{$\{x\}\oplus$}}}}
\put(15,35){{{\tiny{$B(0,r_n)$}}}}
\put(30,12.7){{{\tiny{$C_{u,{{r_n}/2}}$}}}}
\psline[linewidth=.25pt]{->}(3.1,1.5)(2.8,2.4)
\pscircle[linewidth=.5pt](2.5,2.5){.42}
\pscircle[linewidth=.5pt](2.5,2.5){.84}
\psline[linewidth=.5pt]{<->}(1,2.5)(4,2.5)
\psline[linewidth=.5pt]{<->}(2.5,1)(2.5,4.5)
\pscustom[linewidth=.5pt,fillstyle=solid,fillcolor=black]{
\pswedge(2.5,2.5){.42}{30}{90}
}
\end{picture}
\caption{\textit{For $x\in S$ under the conditions stated in Proposition \ref{lejos.R2}, we have that {\mbox{$\{x\}\oplus C_{u,r_n/2}\subset B(x,r_n/2)\subset S$}}.}}
\label{fig:cerca}
\end{center}
\end{figure}
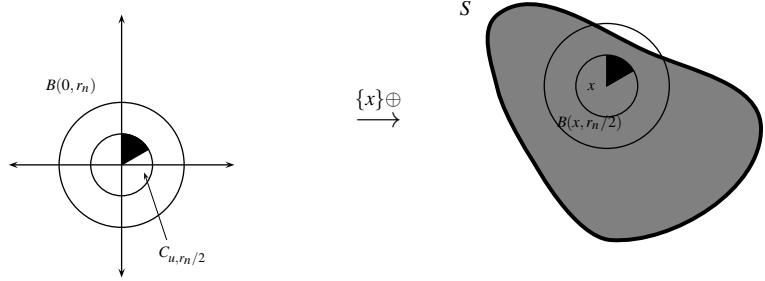


Before proceeding to the definition of unavoidable families of sets for points $x\in S$ with $d(x,\partial S)\leq r_n/2$, we wish to emphasize some aspects of this kind of
families. Recall that for points which lie far away
from the boundary we have proved that it is enough to consider circular sectors with radius $r_n/2$ and central angle $\pi/3$. Using the same argument
for points $x\in S$ such that $\rho=d(x,\partial S)\leq r_n/2$ we only could infer that $B(x,\rho)\subset S$ and hence the lower bound for the probability of these
circular sectors would be of order $\rho^2$. However we can find larger unavoidable sets and improve this bound. To see this, assume without loss
of generality that $x=0$ and divide $B(0,r)$ into a finite number of sectors $C_{u,r}^\theta$ with $\theta>0$. Then for fixed $u$,


\begin{equation}\label{def:U}
U=\bigcap_{y\in C_{u,r}^{\theta}}B(y,r)
\end{equation}
is the largest set contained in $B(y,r)$ for all $y\in C_{u,r}^{\theta}$. 
The measure of $U$ depends on $\theta$. For example,
if $\theta=\pi/2$ then we divide $B(0,r)$ into two circular
sectors with central angle $\pi$. In that case, it can be easily
proved that $U=\{0\}$. Smaller values of $\theta$ result in larger
sets $U$. In particular, Lemma \ref{cont.R2} shows that, fixed $\theta=\pi/6$, the
set in (\ref{def:U}) contains at least one circular sector
with central angle $\pi/3$. In Proposition \ref{cerca.R2} we
show that for points $x\in S$ with $\rho=d(x,\partial S)\leq r_n/2$
and $\theta=\pi/6$ we can give a lower bound for $P_X(U)$ of
order $r_n^{1/2}\rho^{3/2}$. Note that this bound is better than
the one we can obtain for circular sectors of $B(x,\rho)$. Hence,
Proposition \ref{cerca.R2} provides the second key result in the proof of
Theorem \ref{ordenR2}. At this point
it is worth discussing some of the properties of the sets
\begin{equation}\label{tri}
\bigcap_{y\in C_{u,r}}B(y,r),\textnormal{ with } u\in \mathbb{S}_2,\textnormal{ and } r>0.
\end{equation}
As we show in Lemma \ref{triangulo} below, these sets
are known in the literature as Reuleaux triangle\index{Reuleaux triangle},
see Figure \ref{fig:cw2}. They solve the problem
of finding unavoidable families of large sets for the bidimensional case. One
can be tempted to generalize the idea for the {\mbox{$d$-dimensional}}
case.
However the argument
in $\mathbb{R}^d$ is somewhat different since
it becomes tough to handle with the intersection
in (\ref{tri}) when $d>2$. Note that it is fundamental
not only to define large unavoidable sets but also to measure
them. This causes technical difficulties as
the dimension increases. 

\begin{figure}[htb]
\begin{center}
\setlength{\unitlength}{1mm}
\begin{picture}(50,50)
\psline[linewidth=.5pt]{<->}(0,2.5)(5,2.5)
\psline[linewidth=.5pt]{<->}(2.5,0)(2.5,5)
\pscustom[linewidth=.5pt,fillstyle=solid,fillcolor=gray]{
\psarc(2.5,2.5){1}{15}{75}
\psarc(3.456,2.758){1}{135}{195}
\psarc(2.7588,3.4659){1}{255}{315}
}
\psline[linewidth=.5pt,linestyle=dashed]{->}(2.5,2.5)(4.81,4.81)
\pscircle[linewidth=.5pt](2.5,2.5){1}
\pscircle[linewidth=.5pt](3.456,2.758){1}
\pscircle[linewidth=.5pt](2.7588,3.4659){1}
\psdot[dotsize=2pt 1.5](2.5,2.5)
\psdot[dotsize=2pt 1.5](3.456,2.758)
\psdot[dotsize=2pt 1.5](2.7588,3.4659)
\put(23,23){{{\scriptsize{$0$}}}}
\put(44,46.4){{{\scriptsize{$u$}}}}
\put(26,37){{{\scriptsize{$v_1$}}}}
\put(36,27){{{\scriptsize{$v_2$}}}}
\end{picture}
\caption{\textit{Reuleaux triangle.}}
\label{fig:cw2}
\end{center}
\end{figure}
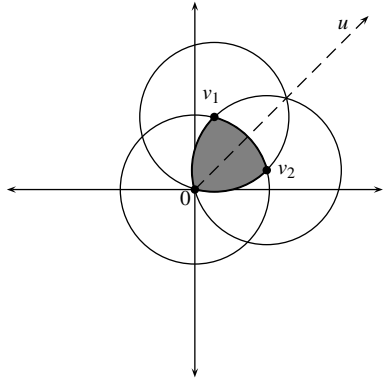
\begin{lemma}\label{triangulo}
Given $u\in \mathbb{S}_2$, we have
\[\bigcap_{y\in C_{u,r}}B(y,r)=B(0,r)\cap B(v_1,r)\cap B(v_2,r),\]
where $v_1=r\mathcal{R}(u)$ and $v_2=r\mathcal{R}^{-1}(u)$, $\mathcal{R}:\mathbb{R}^2\longrightarrow\mathbb{R}^2$ being the counter-clockwise rotation of angle $\pi/6$.
\end{lemma}

\begin{remark}
As previously discussed, the set $B(0,r)\cap B(v_1,r)\cap B(v_2,r)$
in $\mathbb{R}^2$ is the so-called Reuleaux
triangle\index{Reuleaux triangle}. Formally, the Reuleaux triangle
is defined from an equilateral triangle with sides
of length $l$. It is constructed by drawing the arcs from
each polygon vertex of the equilateral triangle between the other
two vertices. Thus, the Reuleaux triangle is the set bounded by these
three arcs. An important property is that it is a set of
constant width $l$, see Figure \ref{fig:cw}. It is known
that the diameter of a set of constant width $l$ is precisely $l$.
See \cite{Benson}, \cite{Croft91}, \cite{Egg58}, and the references
cited therein for a detailed development of these concepts.
\end{remark}

\begin{proof}
It is straightforward to verify
\begin{equation}\label{tritri}
\bigcap_{y\in C_{u,r}}B(y,r)\subset B(0,r)\cap B(v_1,r)\cap B(v_2,r).
\end{equation}
Let us now consider the reverse content.
Let $x\in B(0,r)\cap B(v_1,r)\cap B(v_2,r)$ and $y\in  C_{u,r}$. We need to show that
$\left\|x-y\right\|\leq r$. It follows from (\ref{tritri})
that $y\in B(0,r)\cap B(v_1,r)\cap B(v_2,r)$ and hence, since the
diameter of the Reuleaux triangle\index{Reuleaux triangle}
is $r$, the result holds. \qed 

\end{proof}

%
%
%
%
%

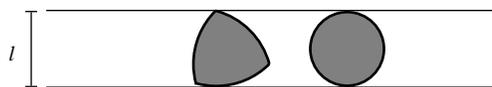
\begin{figure}[htb]
\begin{center}
\setlength{\unitlength}{1mm}
\begin{picture}(70,20)
\pscustom[linewidth=1pt,fillstyle=solid,fillcolor=gray]{
\psarc(2.5,0.5){1}{15}{75}
\psarc(3.456,0.758){1}{135}{195}
\psarc(2.7588,1.4659){1}{255}{315}
}
\pscustom[linewidth=1pt,fillstyle=solid,fillcolor=gray]{
\psarc(4.5,0.96){0.49}{0}{360}
}
\psline[linewidth=.5pt]{-}(0.5,0.45)(6.5,0.45)
\psline[linewidth=.5pt]{-}(0.5,1.471)(6.5,1.471)
\psline[linewidth=.5pt]{|-|}(0.3,1.4659)(0.3,0.45)
\put(0,8){{{\scriptsize{$l$}}}}
\end{picture}
\caption{\textit{Sets of constant width.}}
\label{fig:cw}
\end{center}
\end{figure}


%

We now concentrate on the points $x$ which are close to the
 boundary of $S$. Recall that by points which are close to the boundary of $S$ we mean those $x\in S$ such that $d(x,\partial S)\leq r_n/2$. As previously described, we shall consider in this context unavoidable sets  which are larger than the circular sectors used for points away from $\partial S$. The unavoidable sets $U$ we shortly define guarantee a lower bound for $P_X(U)$ of order $r_n^{1/2}d(x,\partial S)^{3/2}$. Proposition \ref{cerca.R2} makes these ideas precise.

\begin{proposition}\label{cerca.R2}
Under the conditions of Theorem \ref{ordenR2}, {for all} $x\in S$
such that $d(x,\partial S)\leq r_n/2$, there exists a finite family $\mathcal{U}_{x,r_n}$ with $m_2=6$ elements, unavoidable for $\mathcal{E}_{x,r_n}$ and that satisfies
\[P_X(U)\geq L_2r_n^{\frac{1}{2}}d(x,\partial S)^{\frac{3}{2}},\ \ U\in \mathcal{U}_{x,r_n},\]
where the constant $L_2>0$ is independent of $x$.
\end{proposition}

\begin{proof}Let $x\in S$ such that $\rho= d(x,\partial S)\leq r_n/2<\alpha$. Since ${\textnormal{reach}}(\overline{S^c})\geq \alpha$  there
exists a unique point $P_{\Gamma}x\in\partial S$ such that $\rho=\|x- P_{\Gamma}x\|$. The rolling condition ensures the existence of  an unique
unit vector $\eta\equiv\eta(P_{\Gamma}x)$ such that
$B(P_{\Gamma}x-\alpha\eta,\alpha)\subset S$
and therefore, given an unavoidable family $\mathcal{U}_{x,r_n}$,
\begin{equation}\label{des.medida}
P_X(U)\geq\delta\mu(U\cap S)\geq\delta\mu(U\cap B(P_{\Gamma}x-\alpha\eta,\alpha)),\ \ U\in \mathcal{U}_{x,r_n}.
\end{equation}
Note that this simplifies the proof since by (\ref{des.medida}) it follows that we just need to define a suitable family $\mathcal{U}_{x,r_n}$ and bound {\mbox{$\mu(U\cap B(P_{\Gamma}x-\alpha\eta,\alpha))$}} for $U\in\mathcal{U}_{x,r_n}$.
Let us consider a composite function $T$ formed by first applying an orthogonal transformation $\mathcal{O}:\mathbb{R}^2\longrightarrow \mathbb{R}^2$ such that $\mathcal{O}(e_2)=-\eta$ and then applying the translation by the vector $x$, see Figure \ref{transfT}. In particular $T(0)=x$, $T((\alpha-\rho)e_2)=x-(\alpha-\rho)\eta=P_\Gamma x-\alpha\eta$, and $$T(B((\alpha-\rho)e_2,\alpha))=B(P_\Gamma x-\alpha\eta,\alpha).$$

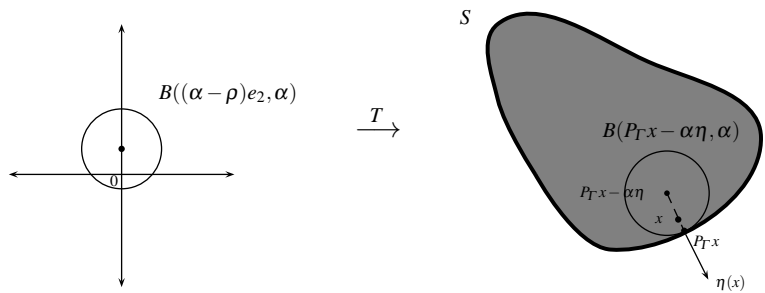
\begin{figure}[htb]
\begin{center}
\setlength{\unitlength}{1mm}
\begin{picture}(115,50)
\pscustom[linewidth=1.5pt,fillstyle=solid,fillcolor=gray]{
\pscurve(8,2.5)(9,1.5)(11,3)(9.5,4)(7.5,4.5)(7.5,3.5)(8,2.5)
}
\pscircle[linewidth=.5pt](9.75,2.25){.57}
\psdot[dotsize=2pt .15](9.75,2.25)
\psdot[dotsize=2pt .5](9.9,1.9)
\psdot[dotsize=2pt .5](9.98,1.75)
\put(96,18.5){{{\tiny{$x$}}}}
\put(86,22){{{\tiny{$P_{\Gamma}x-\alpha\eta$}}}}
\put(101,15.5){{{\tiny{$P_{\Gamma}x$}}}}
\put(104,10){{{\tiny{$\eta(x)$}}}}
\put(89,30){{{\scriptsize{$B(P_{\Gamma}x-\alpha\eta,\alpha)$}}}}
\put(70,45){{{\scriptsize{$S$}}}}
\psline[linewidth=.5pt,linestyle=dashed]{-}(9.75,2.25)(10,1.7)
\psline[linewidth=.5pt]{->}(10,1.7)(10.3,1.1)
\put(56,30){{{{$\longrightarrow$}}}}
\put(58,32){{{\scriptsize{$T$}}}}
\psline[linewidth=.5pt]{<->}(1,2.5)(4,2.5)
\psline[linewidth=.5pt]{<->}(2.5,1)(2.5,4.5)

\pscircle[linewidth=.5pt](2.5,2.84){0.54}
\psdot[dotsize=2pt .5](2.5,2.84)
\put(30,35){{{\scriptsize{$B((\alpha-\rho)e_2,\alpha)$}}}}
\put(23.5,23.46){{{\tiny{$0$}}}}
\end{picture}
\caption{\textit{For the function $T$,  $T(B((\alpha-\rho)e_2,\alpha))=B(P_{\Gamma}x-\alpha\eta,\alpha)$.}}
\label{transfT}
\end{center}
\end{figure}

It can be easily seen that the following result holds.
\begin{lemma}\label{gira}
Let $\mathcal{U}_{0,r}$ be an unavoidable family for $\mathcal{E}_{0,r}$ and let $\mathcal{O}:\mathbb{R}^2\longrightarrow \mathbb{R}^2$\glossary{ $\mathcal{O}$&Orthogonal transformation} be an orthogonal transformation. Then $\{\mathcal{O}(U),\ U\in \mathcal{U}_{0,r}\}$ is also an unavoidable family for $\mathcal{E}_{0,r}$.
\end{lemma}

What Lemma \ref{gira} asserts is that the orthogonal transformation of an unavoidable family for $\mathcal{E}_{0,r_n}$ results in another unavoidable family for $\mathcal{E}_{0,r_n}$. On the other hand, Lemma \ref{traslada2} established that the result of the translation of an unavoidable family for $\mathcal{E}_{0,r_n}$ by the vector $x$ is an unavoidable family for $\mathcal{E}_{x,r_n}$. As an immediate consequence, we obtain that
$\mathcal{U}_{x,r_n}=\{T(U), U\in \mathcal{U}_{0,r_n}\}$
is unavoidable for $\mathcal{E}_{x,r_n}$. Furthermore,
\[\mu(T(U)\cap B(P_{\Gamma}x-\alpha\eta,\alpha))=\mu(U\cap B((\alpha-\rho)e_2,\alpha)),\]
as the Lebesgue measure is invariant under translations and orthogonal transformations. Thus, the problem reduces to defining an unavoidable family $\mathcal{U}_{0,r_n}$ for $\mathcal{E}_{0,r_n}$ and finding a lower bound for {\mbox{$\mu(U\cap B((\alpha-\rho)e_2,\alpha))$}} for all $U\in\mathcal{U}_{0,r_n}$.

Before continuing the proof of Proposition \ref{cerca.R2}, it may be useful to make some comments concerning the measure of the sets $U\cap B((\alpha-\rho)e_2,\alpha)$. Note that when defining unavoidable sets for $\mathcal{E}_{0,r_n}$, the main difficulty in giving a lower bound for {\mbox{$\mu(U\cap B((\alpha-\rho)e_2,\alpha))$}} arises with those points which lie far away in the direction of the vector $-e_2$. In fact,
\begin{equation*}\label{bolas}
\min_{y\in B(0,r_n)}\mu\left(B(y,r_n)\cap B((\alpha-\rho)e_2,\alpha)\right)=\mu(B(-r_ne_2,r_n)\cap B((\alpha-\rho)e_2,\alpha))
\end{equation*}
since $y=-r_ne_2$ represents the point where the distance between the centres of both balls attains its maximum and, as a direct consequence, the intersection its minimum. Recall that, by the definition of unavoidable family, for each $y\in B(0,r_n)$ there exists {\mbox{$U\in \mathcal{U}_{0,r_n}$}} such that {\mbox{$U\subset B(y,r_n)$}}. So, it is more involved to find unavoidable sets $U$ with large enough {\mbox{$\mu(U\cap B((\alpha-\rho)e_2,\alpha))$}} for points close to $-r_ne_2$. This motivates dividing $B(0,r_n)$ into two subsets as follows
$B(0,r_n)=\mathcal{G}_{r_n}\cup \mathcal{F}_{r_n}$
where
\begin{equation*}
\mathcal{G}_{r_n}=\left\{y\in B(0,r_n):\ \left\langle y, e_2\right\rangle\geq -\frac{1}{2}\left\|y\right\|\right\}\mbox{ and }
\mathcal{F}_{r_n}=\left\{y\in B(0,r_n):\ \left\langle y, e_2\right\rangle< -\frac{1}{2}\left\|y\right\|\right\}.
\end{equation*}

Figure \ref{GrFr2} shows the sets $\mathcal{G}_{r_n}$ and $\mathcal{F}_{r_n}$. Roughly speaking, $\mathcal{F}_{r_n}$ contains the points $y\in B(0,r_n)$ for which {\mbox{$B(y,r_n)\cap B((\alpha-\rho)e_2,\alpha)$}} is small. Therefore, the unavoidable sets $U$ in this case should be carefully selected. On the contrary, $\mathcal{G}_{r_n}$ contains the points $y\in B(0,r_n)$ for which {\mbox{$B(y,r_n)\cap B((\alpha-\rho)e_2,\alpha)$}} is larger. 
For these points the sets $U$ can be circular sectors. Proposition \ref{propositionbuenos2} shows that $\mu(U\cap B((\alpha-\rho)e_2,\alpha))$ is then large enough.

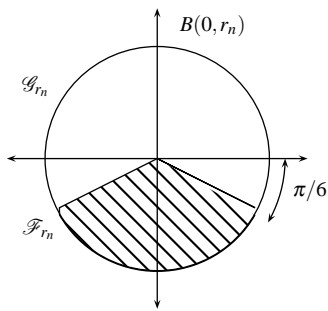
\begin{figure}[htb]
\begin{center}
\setlength{\unitlength}{1mm}
\begin{picture}(50,45)

\psline[linewidth=.5pt]{<->}(0,2)(4,2)
\psline[linewidth=.5pt]{<->}(2,0)(2,4)
\pscustom[linewidth=.5pt,fillstyle=vlines]{
\psline[linewidth=.5pt]{-}(2,2)(3.3,1.346)
\psline[linewidth=.5pt]{-}(2,2)(.7,1.346)
\psarc(2,2){1.5}{210}{330}
}
\psarc[linewidth=.5pt]{<->}(2,2){1.7}{330}{360}
\put(38,15){{{\scriptsize{$\pi/6$}}}}
\put(23,37){{{\scriptsize{$B(0,r_n)$}}}}
\put(2,29){{{\scriptsize{$\mathcal{G}_{r_n}$}}}}
\put(2,10){{{\scriptsize{$\mathcal{F}_{r_n}$}}}}
\pscircle[linewidth=.5pt](2,2){1.5}
\end{picture}
\caption{\textit{$\mathcal{G}_{r_n}$ and $\mathcal{F}_{r_n}$.}}
\label{GrFr2}
\end{center}
\end{figure}

\begin{proposition}\label{propositionbuenos2}
There exists a finite set of unit vectors {\mbox{$\mathcal{W}^{\mathcal{G}}\subset\mathbb{S}_2$}} with $m^{\mathcal{G}}=4$ elements such that, for all $y\in\mathcal{G}_{r_n}$, there
 exists $u\in\mathcal{W}^{\mathcal{G}}$ such that {\mbox{$y\in C_{u,r_n}\subset B(y,r_n)$}} and
\begin{equation*}\label{obj2}
\mu(C_{u,r_n}\cap B((\alpha-\rho)e_2,\alpha))\geq L^\mathcal{G} r_n^{\frac{1}{2}}\rho^{\frac{3}{2}},
\end{equation*}
where $L^\mathcal{G}>0$ is a constant.
\end{proposition}
\begin{proof}Let us consider the set $\mathcal{W}^{\mathcal{G}}=\{(1,0),(-1,0),(1/2,\sqrt{3}/2),(-1/2,\sqrt{3}/2)\}$. It is straightforward to verify, see Figure \ref{fig:GR2}, that
$\mathcal{G}_{r_n}=\bigcup_{u\in\mathcal{W}^{\mathcal{G}}}C_{u,r_n}.$
Therefore, for all $y\in\mathcal{G}_{r_n}$ there exists $u\in\mathcal{W}^{\mathcal{G}}$ such that {\mbox{$y\in C_{u,r_n}$}}. By Lemma \ref{cont.R2} it follows
that $C_{u,r_n}\subset B(y,r_n)$. It remains to find a lower bound for $C_{u,r_n}\cap B((\alpha-\rho)e_2,\alpha)$ for $u\in\mathcal{W}^{\mathcal{G}}$. Note that at least
half of the set $C_{u,r_n}$ is contained in the halfplane $H_0=\{x=(x_1,x_2) \in\mathbb{R}^2:\ x_2\geq 0\}$ and hence it is sufficient for our purposes to concentrate on $C_{u,r_n}\cap H_0$.

\begin{figure}[htb]
\begin{center}
\setlength{\unitlength}{1mm}
\begin{picture}(130,45)
\psline[linewidth=.5pt]{<->}(4.2,2)(8.8,2)
\psline[linewidth=.5pt]{<->}(6.5,0)(6.5,4)
\psarc[linewidth=.5pt](6.5,2){1.5}{330}{210}
\psline[linewidth=.5pt]{-}(6.5,2)(7.79,1.25)
\psline[linewidth=.5pt]{-}(6.5,2)(5.21,1.25)

\psarc[linewidth=.5pt]{<->}(6.5,2){1.7}{330}{360}
\put(82,13){{{\scriptsize{$\pi/6$}}}}

\pscustom[linewidth=.5pt,hatchwidth=.2pt,fillstyle=hlines]{
\psline[linewidth=.5pt]{-}(6.5,2)(7.78,1.25)
\psarc(6.5,2){1.5}{330}{30}
\psline[linewidth=.5pt]{-}(7.78,2.75)(6.5,2)
}
\pscustom[linewidth=.5pt,hatchwidth=.2pt,fillstyle=vlines]{
\psline[linewidth=.5pt]{-}(6.5,2)(7.78,2.75)
\psarc(6.5,2){1.5}{30}{90}
\psline[linewidth=.5pt]{-}(6.5,3.5)(6.5,2)
}
\pscustom[linewidth=.5pt,hatchwidth=.2pt,fillstyle=hlines]{
\psline[linewidth=.5pt]{-}(6.5,2)(6.5,3.5)
\psarc(6.5,2){1.5}{90}{150}
\psline[linewidth=.5pt]{-}(5.22,2.75)(6.5,2)
}
\pscustom[linewidth=.5pt,hatchwidth=.2pt,fillstyle=vlines]{
\psline[linewidth=.5pt]{-}(6.5,2)(5.22,2.75)
\psarc(6.5,2){1.5}{150}{210}
\psline[linewidth=.5pt]{-}(5.22,1.25)(6.5,2)
}
\psline[linewidth=1pt]{->}(6.5,2)(8,2)
\psline[linewidth=1pt]{->}(6.5,2)(5,2)
\psline[linewidth=1pt]{->}(6.5,2)(7.23,3.29)
\psline[linewidth=1pt]{->}(6.5,2)(5.77,3.29)
\psline[linewidth=.5pt]{-}(6.5,2)(6.5,3.5)
\put(81,21.5){{{\scriptsize{$(1,0)$}}}}
\put(40.5,21.5){{{\scriptsize{$(-1,0)$}}}}
\put(72,35){{{\scriptsize{$\left(\frac{1}{2},\frac{\sqrt{3}}{2}\right)$}}}}
\put(43,35){{{\scriptsize{$\left(-\frac{1}{2},\frac{\sqrt{3}}{2}\right)$}}}}
\end{picture}
\caption{\textit{Unit vectors  $\mathcal{W}^{\mathcal{G}}=\{(1,0),(-1,0),(1/2,\sqrt{3}/2),(-1/2,\sqrt{3}/2)\}$ and $C_{u,r_n}$, for $u\in\mathcal{W}^{\mathcal{G}}$.}}
\label{fig:GR2}
\end{center}
\end{figure}
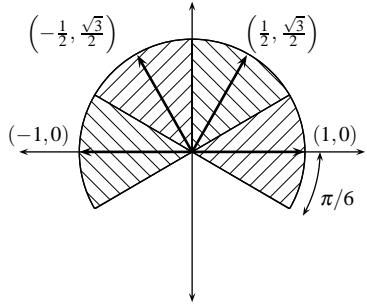

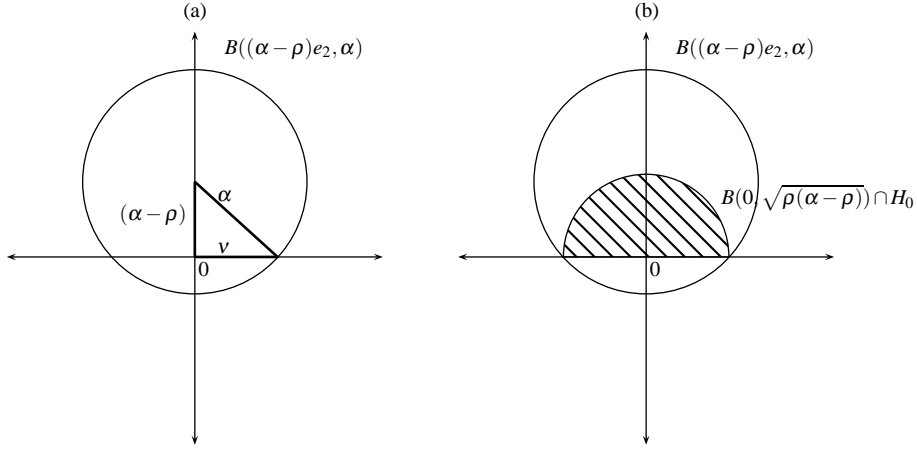
\begin{figure}
\begin{center}
\setlength{\unitlength}{1mm}
\begin{picture}(120,60)

\psline[linewidth=.5pt]{<->}(0,2.5)(5,2.5)
\psline[linewidth=.5pt]{<->}(2.5,0)(2.5,5.5)
\pscircle[linewidth=.5pt](2.5,3.5){1.5}
\psline[linewidth=1pt]{-}(2.5,3.5)(3.6,2.5)
\psline[linewidth=1pt]{-}(2.5,3.5)(2.5,2.5)
\psline[linewidth=1pt]{-}(2.5,2.5)(3.6,2.5)

\put(25.5,22.5){{{\scriptsize{$0$}}}}

\put(15,30){{{\scriptsize{$(\alpha-\rho)$}}}}
\put(28,25.5){{{\scriptsize{$\nu$}}}}
\put(28,32.5){{{\scriptsize{$\alpha$}}}}

\psline[linewidth=.5pt]{<->}(6,2.5)(11,2.5)
\psline[linewidth=.5pt]{<->}(8.5,0)(8.5,5.5)
\pscircle[linewidth=.5pt](8.5,3.5){1.5}

\pscustom[linewidth=.5pt,fillstyle=vlines]{
\psarc[linewidth=.5pt](8.5,2.5){1.1}{0}{180}
\psline{-}(7.4,2.5)(9.6,2.5)
}

\put(85.5,22.5){{{\scriptsize{$0$}}}}

\put(29,52){{{\scriptsize{$B((\alpha-\rho)e_2,\alpha)$}}}}
\put(89,52){{{\scriptsize{$B((\alpha-\rho)e_2,\alpha)$}}}}
\put(95,32){{{\scriptsize{$B(0,\sqrt{\rho(\alpha-\rho)})\cap H_0$}}}}

\put(23.5,57){{{\scriptsize{(a)}}}}
\put(83.5,57){{{\scriptsize{(b)}}}}
\end{picture}
\caption{\textit{(a) $\nu=\sqrt{\rho(2\alpha-\rho)}$. (b) $ B(0,\nu)\cap H_0 \subset B((\alpha-\rho)e_2,\alpha)$.}}
\label{fig:pita2}
\end{center}
\end{figure}

Let $\nu=\sqrt{\rho(2\alpha-\rho)}$. By the Pythagorean theorem, it is straightforward
to see that $\nu$ represents the distance to the origin from the points such that $\partial B((\alpha-\rho)e_2,\alpha)$ intersects the axis $OX$, see Figure \ref{fig:pita2}. It is also easy to
show that $B(0,\nu)\cap H_0\subset B((\alpha-\rho)e_2,\alpha)$. Therefore, for $u\in\mathcal{W}^{\mathcal{G}}$, $C_{u,\tau_n}\cap H_0\subset C_{u,r_n}\cap B((\alpha-\rho)e_2,\alpha)$
where $\tau_n=\min(\nu,r_n)$. This yields,
\begin{equation*}
\mu(C_{u,r_n}\cap B((\alpha-\rho)e_2,\alpha))\geq\mu(C_{u,\tau_n}\cap H_0)\geq\frac{1}{2}\mu(C_{u,\tau_n})=\frac{\pi}{12}\tau_n^2\geq \frac{\pi}{12}r_n^{1/2}\rho^{3/2}.
\end{equation*}
This completes the proof of Proposition \ref{propositionbuenos2}, with $L^\mathcal{G}=\pi/{12}>0$ constant. \qed 

\end{proof}

In view of Proposition \ref{propositionbuenos2} we define the family
$\mathcal{U}_{0,r_n}^\mathcal{G}=\{C_{u,r_n}, u\in\mathcal{W}^{\mathcal{G}}\},$
formed by $m^\mathcal{G}=4$ elements. We now turn to the points in $\mathcal{F}_{r_n}$. The aim is to define for those points a finite family $\mathcal{U}^{\mathcal{F}}_{0,r_n}$, such that, for all $y\in\mathcal{F}_{r_n}$, there exists $U\in\mathcal{U}^{\mathcal{F}}_{0,r_n}$ that satisfies $U\subset B(y,r_n)$ and
\begin{equation}\label{bufR2}
\mu(U\cap B((\alpha-\rho)e_2,\alpha))\geq L^{\mathcal{F}} r_n^{\frac{1}{2}}\rho^{\frac{3}{2}},\ \ \forall U\in \mathcal{U}^{\mathcal{F}}_{0,r_n}.
\end{equation}
At this point, it may be useful to make some comments concerning the main differences between $\mathcal{G}_{r_n}$ and $\mathcal{F}_{r_n}$.
One might be tempted to proceed as before for $\mathcal{F}_{r_n}$ and define the set of unit vectors
$\mathcal{W}^\mathcal{F}=\{(-1/2,-\sqrt{3}/2),(1/2,-\sqrt{3}/2)\}.$
Again we would have that, see Figure \ref{fig:FR2} (a), $\mathcal{F}_{r_n}=\bigcup_{u\in\mathcal{W}^{\mathcal{F}}}C_{u,r_n}.$


\begin{figure}[htb]
\begin{center}
\setlength{\unitlength}{1mm}
\begin{picture}(120,50)
\psline[linewidth=.5pt]{<->}(0.2,2)(4.8,2)
\psline[linewidth=.5pt]{<->}(2.5,0)(2.5,4.5)
\psline[linewidth=.5pt]{-}(2.5,2)(3.79,1.25)
\psline[linewidth=.5pt]{-}(2.5,2)(1.21,1.25)
\pscustom[linewidth=.5pt,hatchwidth=.2pt,fillstyle=hlines]{
\psline[linewidth=.5pt]{-}(2.5,2)(2.5,.5)
\psarc(2.5,2){1.5}{270}{330}
\psline[linewidth=.5pt]{-}(3.78,1.25)(2.5,2)
}
\pscustom[linewidth=.5pt,hatchwidth=.2pt,fillstyle=vlines]{
\psline[linewidth=.5pt]{-}(2.5,2)(1.22,1.25)
\psarc(2.5,2){1.5}{210}{270}
\psline[linewidth=.5pt]{-}(2.5,.5)(2.5,2)
}
\psline[linewidth=1pt]{->}(2.5,2)(3.23,0.71)
\psline[linewidth=1pt]{->}(2.5,2)(1.77,0.71)
\psline[linewidth=.5pt]{-}(2.5,2)(2.5,3.5)
\put(34,4){{{\scriptsize{$\left(\frac{1}{2},-\frac{\sqrt{3}}{2}\right)$}}}}
\put(0,4){{{\scriptsize{$\left(-\frac{1}{2},-\frac{\sqrt{3}}{2}\right)$}}}}

\psline[linewidth=.5pt]{<->}(6.2,2)(10.8,2)
\psline[linewidth=.5pt]{<->}(8.5,0)(8.5,4.5)
\psline[linewidth=.5pt]{-}(8.5,2)(9.79,1.25)
\psline[linewidth=.5pt]{-}(8.5,2)(7.21,1.25)
\pscustom[linewidth=.5pt,fillstyle=solid,fillcolor=gray]{
\psline[linewidth=.5pt]{-}(8.5,2)(8.5,1.3)
\psline[linewidth=.5pt]{-}(8.5,1.3)(9.71,1.3)
\psline[linewidth=.5pt]{-}(9.71,1.3)(9.71,2)
\psline[linewidth=.5pt]{-}(9.71,2)(8.5,2)
}
\pscustom[linewidth=.5pt,hatchwidth=.2pt,fillstyle=hlines]{
\psline[linewidth=.5pt]{-}(8.5,2)(8.5,.5)
\psarc(8.5,2){1.5}{270}{330}
\psline[linewidth=.5pt]{-}(9.78,1.25)(8.5,2)
}
\pscustom[linewidth=.5pt,hatchwidth=.2pt,fillstyle=vlines]{
\psline[linewidth=.5pt]{-}(8.5,2)(7.22,1.25)
\psarc(8.5,2){1.5}{210}{270}
\psline[linewidth=.5pt]{-}(8.5,.5)(8.5,2)
}
\psline[linewidth=1pt]{->}(8.5,2)(9.23,0.71)
\psline[linewidth=1pt]{->}(8.5,2)(7.77,0.71)
\psline[linewidth=.5pt]{-}(8.5,2)(8.5,3.5)
\pscircle[linewidth=.5pt](8.5,2.7){1.4}
\put(89,42){{{\scriptsize{$B((\alpha-\rho)e_2,\alpha)$}}}}
\put(94,4){{{\scriptsize{$\left(\frac{1}{2},-\frac{\sqrt{3}}{2}\right)$}}}}
\put(60,4){{{\scriptsize{$\left(-\frac{1}{2},-\frac{\sqrt{3}}{2}\right)$}}}}
\put(99,15){{{\scriptsize{$\rho$}}}}
\put(89,21){{{\scriptsize{$\sqrt{3}\rho$}}}}
\put(23.5,50){{{\scriptsize{(a)}}}}
\put(83.5,50){{{\scriptsize{(b)}}}}

\end{picture}
\caption{\textit{(a) $\mathcal{W}^{\mathcal{F}}=\{(-1/2,-\sqrt{3}/2),(1/2,-\sqrt{3}/2)\}$ and $C_{u,r_n}$, for $u\in\mathcal{W}^{\mathcal{F}}$. (b) For $u\in\mathcal{W}^{\mathcal{F}}$, $C_{u,r_n}\cap B((\alpha-\rho)e_2,\alpha)$ is contained in the rectangle of height $\rho$ and base $\sqrt{3}\rho$.}}
\label{fig:FR2}
\end{center}
\end{figure}
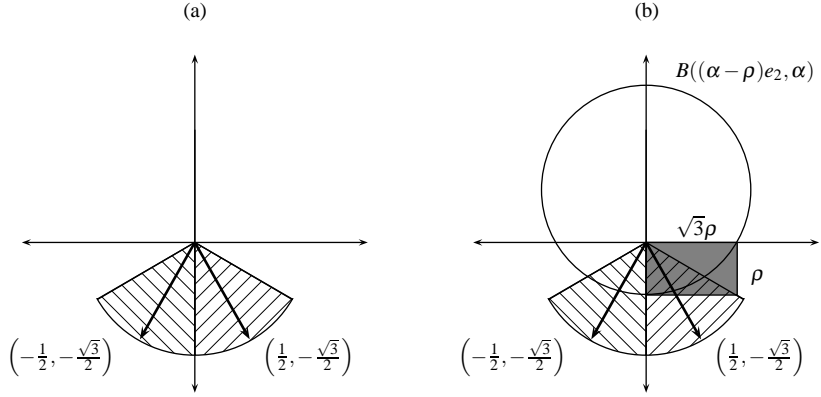

If we repeat the sketch of the proof for $\mathcal{G}_{r_n}$ and define $U$ to be the circular sectors $C_{u,r_n}$ for $u\in\mathcal{W}^{\mathcal{F}}$, we could
no longer guarantee the lower bound in (\ref{bufR2}). Note that the intersection $C_{u,r_n}\cap B((\alpha-\rho)e_2,\alpha)$ for $u\in\mathcal{W}^{\mathcal{F}}$ is considerably
smaller than for $u\in\mathcal{W}^{\mathcal{G}}$. In fact, it can be easily proved that, for $u\in\mathcal{W}^{\mathcal{F}}$,
$\mu(C_{u,r_n}\cap B((\alpha-\rho)e_2,\alpha))\leq \sqrt{3}\rho^2,$
as it is shown in Figure \ref{fig:FR2} (b). Therefore, we need to consider different sets $U$. We have previously discussed the possibility of defining unavoidable sets, larger than circular sectors. For a fixed unit vector $u$,
\begin{equation}\label{uu}
U=\bigcap_{y\in C_{u,r_n}}B(y,r_n)
\end{equation}
is the largest set such that $U\subset B(y,r_n)$ for all $y\in C_{u,r_n}$.
\begin{figure}[h]
\begin{center}
\setlength{\unitlength}{1mm}
\begin{picture}(120,60)
\psline[linewidth=.5pt]{<->}(0,2.5)(5,2.5)
\psline[linewidth=.5pt]{<->}(2.5,0)(2.5,5.5)
\pscircle[linewidth=.5pt](2.5,3.5){1.5}
\psline[linewidth=1pt]{->}(2.5,2.5)(3.23,1.21)
\pscustom[linewidth=.5pt,hatchwidth=.2pt,fillstyle=hlines]{
\psline[linewidth=.5pt]{-}(2.5,2.5)(2.5,1)
\psarc(2.5,2.5){1.5}{270}{330}
\psline[linewidth=.5pt]{-}(3.78,1.75)(2.5,2.5)
}

\psline[linewidth=.5pt]{<->}(6,2.5)(11,2.5)
\psline[linewidth=.5pt]{<->}(8.5,0)(8.5,5.5)
\pscircle[linewidth=.5pt](8.5,3.5){1.5}
\pscustom[linewidth=.5pt,hatchwidth=.2pt,fillstyle=hlines]{
\psarc(8.5,1){1.5}{30}{90}
\psarc(9.79,1.75){1.5}{150}{210}
\psarc(8.5,2.5){1.5}{270}{330}
}
\psline[linewidth=1pt]{->}(8.5,2.5)(9.23,1.21)
\pscustom[linewidth=.5pt]{
\psline[linewidth=.5pt]{-}(8.5,2.5)(8.5,1)
\psarc(8.5,2.5){1.5}{270}{330}
\psline[linewidth=.5pt]{-}(9.78,1.75)(8.5,2.5)
}
\put(29,52){{{\scriptsize{$B((\alpha-\rho)e_2,\alpha)$}}}}
\put(89,52){{{\scriptsize{$B((\alpha-\rho)e_2,\alpha)$}}}}

\put(33,8){{{\scriptsize{$u=\left(\frac{1}{2},-\frac{\sqrt{3}}{2}\right)$}}}}
\put(94,8){{{\scriptsize{$u=\left(\frac{1}{2},-\frac{\sqrt{3}}{2}\right)$}}}}

\put(23.5,57){{{\scriptsize{(a)}}}}
\put(83.5,57){{{\scriptsize{(b)}}}}
\end{picture}
\caption{\textit{(a) $C_{u,r_n}$ with $u=(1/2,-\sqrt{3}/2)$. (b) $\bigcap_{y\in C_{u,r_n}}B(y,r_n)$.}}
\label{fig:compara}
\end{center}
\end{figure}
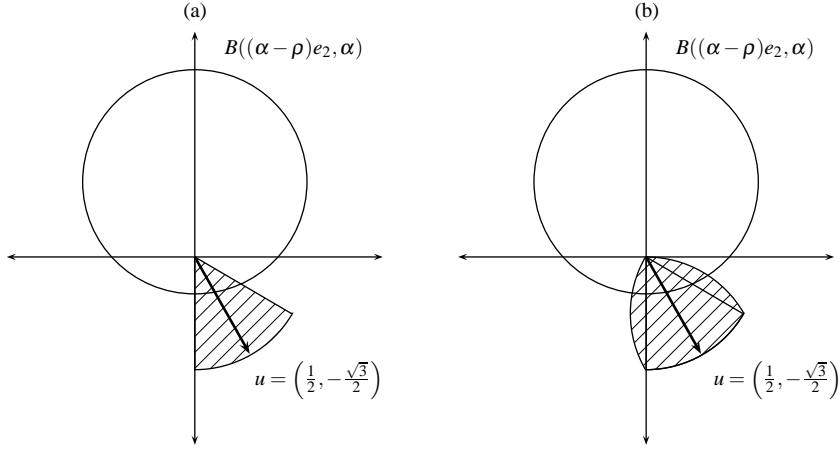
Figure \ref{fig:compara} shows $C_{u,r_n}$, for an $u\in\mathcal{W}^\mathcal{F}$ and the corresponding set $U$ defined in (\ref{uu}). Observe that $U\cap B((\alpha-\rho)e_2,\alpha)$ is clearly larger than $C_{u,r_n}\cap B((\alpha-\rho)e_2,\alpha)$. The difference between both intersections will play a fundamental role in obtaining the lower bound in (\ref{bufR2}). In fact, it is not necessary to consider the whole $U$ as defined in (\ref{uu}). For our purposes it is sufficient to measure a portion of $U\cap B((\alpha-\rho)e_2,\alpha)$. We shall consider sets as the one represented in gray in Figure \ref{midoesto}. Its measure is large enough to satisfy (\ref{bufR2}). We give the precise definition of this kind of sets in Proposition \ref{propositionmalos2}. This solves the problem for the points in $\mathcal{F}_{r_n}$.

\begin{figure}
\begin{center}
\setlength{\unitlength}{1mm}
\begin{picture}(50,60)
\pscustom[linewidth=.5pt,fillstyle=solid,fillcolor=gray]{
\psline(2.5,2.5)(2.5,2.25)(3.3,2.25)
\psarc(2.5,1){1.5}{58}{90}
}
\pscustom[linewidth=.5pt,hatchwidth=.2pt,fillstyle=hlines]{
\psarc(2.5,1){1.5}{30}{90}
\psarc(3.79,1.75){1.5}{150}{210}
\psarc(2.5,2.5){1.5}{270}{330}
}
\psline[linewidth=1pt]{->}(2.5,2.5)(3.23,1.21)
\psline[linewidth=.5pt]{<->}(0,2.5)(5,2.5)
\psline[linewidth=.5pt]{<->}(2.5,-0.25)(2.5,5.5)
\pscircle[linewidth=.5pt](2.5,3.5){1.5}
\put(33,8){{{\scriptsize{$u=\left(\frac{1}{2},-\frac{\sqrt{3}}{2}\right)$}}}}
\end{picture}
\caption{\textit{The dashed area corresponds to $U=\bigcap_{y\in C_{u,r_n}}B(y,r_n)$ with $u=(1/2,-\sqrt{3}/2)$.}}
\label{midoesto}
\end{center}
\end{figure}
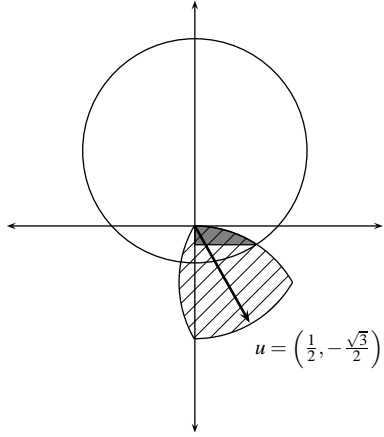

\begin{proposition}\label{propositionmalos2}
There exists a finite family of sets $\mathcal{U}^{\mathcal{F}}_{0,r_n}$ with $m^{\mathcal{F}}=2$ elements such that, for all $y\in\mathcal{F}_{r_n}$, there exists $U\in\mathcal{U}^{\mathcal{F}}_{0,r_n}$ such that {\mbox{$U\subset B(y,r_n)$}} and
\begin{equation*}
\mu(U\cap B((\alpha-\rho)e_2,\alpha))\geq L^\mathcal{F} r_n^{\frac{1}{2}}\rho^{\frac{3}{2}},
\end{equation*}
with $L^\mathcal{F}>0$ a constant.
\end{proposition}

\begin{proof}
First, let us consider the set
$B((\alpha-\rho)e_2,\alpha)\cap{B}(-r_ne_2,r_n)$,
which corresponds to the intersection between two balls of radii $\alpha$ and $r_n$, respectively, being  $\alpha+r_n-\rho$ the distance between their centres, see Figure \ref{fig:Ly2} (a).
The values of $h_1$, $h_2$ and $\lambda$ in Figure \ref{fig:Ly2} (b) can be deduced from the Pythagorean theorem. They satisfy the following equations
\begin{equation*}
\left\{\begin{array}{l}
(r_n-h_1)^2+\lambda^2=r_n^2,\\
(\alpha-h_2)^2+\lambda^2=\alpha^2,\\
h_1+h_2=\rho.
\end{array}\right.
\end{equation*}
By solving the system,
\begin{equation*}
h_1=\frac{\rho(2\alpha-\rho)}{2(\alpha+r_n-\rho)},\ \ \ h_2=\frac{\rho(2r_n-\rho)}{2(\alpha+r_n-\rho)}\ \ \textnormal{ , and }\ \ \lambda=\sqrt{2r_nh_1-h_1^2}.
\end{equation*}
We now define the set
\begin{equation}
\label{DefCh1}
\mathcal{C}(h_1)=\{x\in\mathbb{R}^2:\ -h_1\leq\left\langle x,e_2\right\rangle\leq 0\}\cap B(-r_ne_2,r_n).
\end{equation}
Lemma \ref{Ch1} provides a lower bound for the measure of $\mathcal{C}(h_1)$.
\begin{figure}[h]
\begin{center}
\setlength{\unitlength}{1mm}
\begin{picture}(140,60)
\psarc[linewidth=.5pt,fillstyle=vlines*,fillcolor=gray](2.5,1.3){1.2}{50}{130}
\psarc[linewidth=.5pt,fillstyle=vlines](2.5,3.5){1.5}{239}{301}
\psline[linewidth=.5pt]{<->}(0,2.5)(5,2.5)
\psline[linewidth=.5pt]{<->}(2.5,-0.25)(2.5,5.5)
\pscircle[linewidth=.5pt](2.5,1.3){1.2}
\pscircle[linewidth=.5pt](2.5,3.5){1.5}
\psline[linewidth=.5pt]{-}(2.5,3.5)(1.13,4.05)
\psline[linewidth=.5pt]{-}(2.5,1.3)(1.37,0.98)
\put(25.5,22.5){{{\scriptsize{$0$}}}}
\put(20,38){{{\scriptsize{$\alpha$}}}}
\put(21,10){{{\scriptsize{$r_n$}}}}
\put(23.5,57){{{\scriptsize{(a)}}}}
\psline[linewidth=.5pt]{<->}(6,2.5)(11,2.5)
\psline[linewidth=.5pt]{<->}(8.5,-0.25)(8.5,5.5)
\pscircle[linewidth=.5pt](8.5,1.3){1.2}
\pscircle[linewidth=.5pt](8.5,3.5){1.5}
\psline[linewidth=1pt]{-}(8.5,3.5)(7.75,2.23)
\psline[linewidth=1pt]{-}(8.5,1.3)(7.75,2.23)
\psline[linewidth=1pt]{-}(7.75,2.23)(8.5,2.23)
\psline[linewidth=1pt]{-}(8.5,1.3)(8.5,3.5)
\psline[linewidth=.5pt,linestyle=dashed]{-}(6,2.23)(11,2.23)
\psline[linewidth=.5pt,linestyle=dashed]{-}(6,2)(11,2)
\put(85.5,22.5){{{\scriptsize{$0$}}}}
\put(78,27.5){{{\scriptsize{$\alpha$}}}}
\put(79,15){{{\scriptsize{$r_n$}}}}
\put(83.5,57){{{\scriptsize{(b)}}}}
\put(82,22.55){{{\scriptsize{$\lambda$}}}}
\psline[linewidth=.5pt]{|-|}(11.2,2)(11.2,2.23)
\psline[linewidth=.5pt]{-|}(11.2,2.23)(11.2,2.5)
\put(113,19.5){{{\tiny{$h_2=\frac{\rho(2r_n-\rho)}{2(\alpha+r_n-\rho)}$}}}}
\put(113,24){{{\tiny{$h_1=\frac{\rho(2\alpha-\rho)}{2(\alpha+r_n-\rho)}$}}}}
\psline[linewidth=.5pt]{|-|}(13.5,2)(13.5,2.5)
\put(136,22){{{\scriptsize{$\rho$}}}}
\end{picture}
\caption{\textit{(a) The dashed area corresponds to $B((\alpha-\rho)e_2,\alpha)\cap{B}(-r_ne_2,r_n)$. In gray $\mathcal{C}(h_1)$.  (b) Values of $h_1$, $h_2$ and $\lambda$.}}
\label{fig:Ly2}
\end{center}
\end{figure}
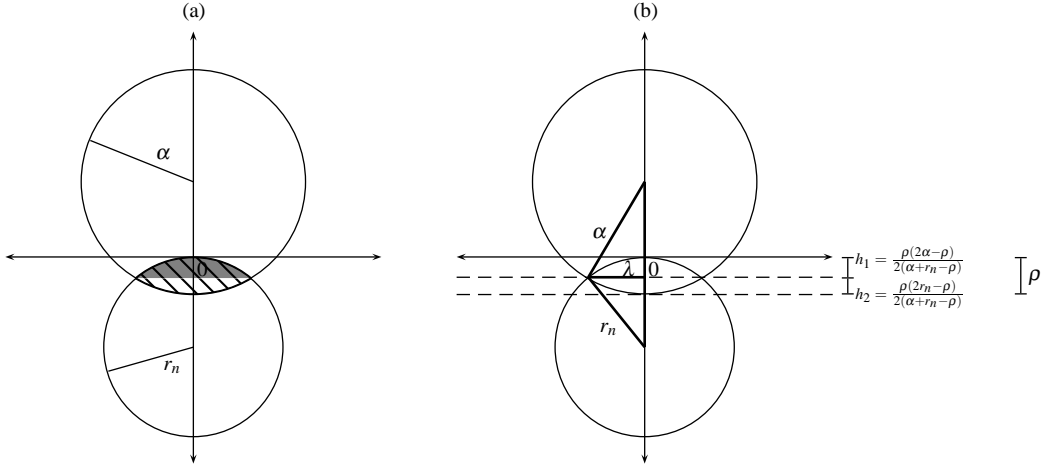
\begin{lemma}\label{Ch1}
Given the previous set $\mathcal{C}(h_1)$, then
\[\mu(\mathcal{C}(h_1))\geq\frac{\sqrt{2}}{3}r_n^{\frac{1}{2}}\rho^{\frac{3}{2}}.\]
\end{lemma}

\begin{proof}
We have that
\begin{equation}\label{integral}
\mu(\mathcal{C}(h_1))=\int_{0}^{h_1}{2\sqrt{2r_ny-y^2}dy}.
\end{equation}
For $y\in[0,h_1]$ we have that $y\leq r_n$, since by construction $h_1\leq\rho$ and by assumption $\rho\leq r_n/2$. Hence, $2r_ny-y^2\geq r_ny$ and
\[\mu(\mathcal{C}(h_1))\geq\int_{0}^{h_1}{2\sqrt{r_ny}dy}=\frac{4}{3}r_n^{\frac{1}{2}}h_1^{\frac{3}{2}}.\]
Moreover, $h_1\geq\rho/2$, since $r_n\leq \alpha$ and this completes the proof. \qed 

\end{proof}

\begin{remark}
Note that the exact value of the integral in (\ref{integral}) can be explicitly computed since it coincides with the area of the circular segment defined by the chord that joins the intersection points of $B((\alpha-\rho)e_2,\alpha)\cap{B}(-r_ne_2,r_n)$. Thus,
\[\mu(\mathcal{C}(h_1))=r_n^2\arccos\left(\frac{r_n-h_1}{r_n}\right)-(r_n-h_1)\sqrt{2r_nh_1-h_1^2}.\]
\end{remark}

 So, we have defined the set $\mathcal{C}(h_1)$, whose measure verifies the statement of Proposition \ref{propositionmalos2}. Next lemma shows that $\mathcal{C}(h_1)$ is contained in  $B((\alpha-\rho)e_2,\alpha)$.

\begin{lemma}\label{Ch1contR2}
\[\mathcal{C}(h_1)\subset B((\alpha-\rho)e_2,\alpha).\]
\end{lemma}
\begin{proof}
Let $x\in\mathcal{C}(h_1)$.
\begin{eqnarray}\label{aux}
\left\|x-(\alpha-\rho)e_2\right\|^2&=&\left\|x\right\|^2+(\alpha-\rho)^2-2(\alpha-\rho)\left\langle x,e_2\right\rangle.
\end{eqnarray}
By definition, $x\in B(-r_ne_2,r_n)$ and therefore 
$\left\|x\right\|^2\leq -2r_n\left\langle x,e_2\right\rangle.$
Furthermore, by definition, $\left\langle x,e_2\right\rangle\geq -h_1$. Turning to (\ref{aux}) we get
\begin{eqnarray*}
\left\|x-(\alpha-\rho)e_2\right\|^2&\leq&2r_nh_1+(\alpha-\rho)^2+2(\alpha-\rho)h_1\\
&=&\rho(2\alpha-\rho)+(\alpha-\rho)^2=\alpha^2.
\end{eqnarray*} \qed

\end{proof}

It follows from Lemmas \ref{Ch1} and \ref{Ch1contR2} that
\begin{equation}\label{Csirve}
\mu(\mathcal{C}(h_1)\cap B((\alpha-\rho)e_2,\alpha))\geq L r_n^{\frac{1}{2}}\rho^{\frac{3}{2}}.
\end{equation}
In order to complete the proof, it remains to define the family $\mathcal{U}^\mathcal{F}_{0,r_n}$ mentioned in the statement
of Proposition \ref{propositionmalos2}. In view of (\ref{Csirve}), it seems natural to divide $\mathcal{C}(h_1)$.
We denote $Q_1=\{x=(x_1,x_2) \in\mathbb{R}^2:\ x_1\geq 0\}$ and $Q_2=\{x=(x_1,x_2) \in\mathbb{R}^2:\ x_1\leq 0\}$. Then,
$\mathcal{F}_{r_n}=(Q_1\cap\mathcal{F}_{r_n})\cup (Q_{2}\cap\mathcal{F}_{r_n})$
and, in the same manner, $\mathcal{C}(h_1)=(Q_1\cap\mathcal{C}(h_1))\cup (Q_{2}\cap\mathcal{C}(h_1)).$

\begin{lemma}\label{contenido2} For all $y\in Q_i\cap\mathcal{F}_{r_n}$ we have that
\[Q_i\cap\mathcal{C}(h_1)\subset B(y,r_n), \ \ i=1,2.\]
\end{lemma}
\begin{proof}
Let $x\in Q_1\cap\mathcal{C}(h_1)$. First, it can be easily proved that
$Q_1\cap\mathcal{F}_{r_n}=C_{u,r_n},$
with $u=(1/2,-\sqrt{3}/2) $. What we need to prove is
$x\in \bigcap_{y\in C_{u,r_n}}B(y,r_n).$
It follows from Lemma \ref{triangulo} that
\[\bigcap_{y\in C_{u,r_n}}B(y,r_n)=B(0,r_n)\cap B(v_1,r_n)\cap B(v_2,r_n),\]
where $v_1=r_n\mathcal{R}(u)=r_n\left(\sqrt{3}/2, -1/2\right)$ and $v_2=r_n\mathcal{R}^{-1}(u)=-r_ne_2$. We have by definition that $x\in B(v_2,r_n)$. Moreover,
$\left\|x\right\|^2\leq\lambda^2+h_1^2= 2r_nh_1\leq r_n^2,$
since $h_1\leq\rho\leq r_n/2$. Note that the last inequality justifies the choice of $\rho\leq r_n/2$. And,
\[\left\|x-v_1\right\|^2=\left(x_1-\frac{\sqrt{3}r_n}{2}\right)^2+\left(x_2+\frac{r_n}{2}\right)^2\leq \left(\frac{\sqrt{3}r_n}{2}\right)^2+\left(\frac{r_n}{2}\right)^2= r_n^2,\]
since $0\leq x_1\leq\lambda\leq {\sqrt{3}r_n}/2$ and $-h_1\leq x_2\leq 0$, where $h_1\leq \rho\leq r_n/2$. Thus, we have shown that
$x\in B(0,r_n)\cap B(v_1,r_n)\cap B(v_2,r_n)$
and the lemma is proved for $Q_1\cap\mathcal{C}(h_1)$.
The proof for $Q_2\cap\mathcal{C}(h_1)$ is analogous. \qed 

\end{proof}

In view of the previous results we define the family
$\mathcal{U}_{0,r_n}^\mathcal{F}=\{Q_i\cap\mathcal{C}(h_1),\ i=1,2\},$
formed by $m^\mathcal{F}=2$ elements. It follows from Lemma \ref{contenido2} that, for all $y\in\mathcal{F}_{r_n}$, there exists $i\in\{1,2\}$ such that $Q_i\cap\mathcal{C}(h_1)\subset B(y,r_n)$. Moreover, by Lemma \ref{Ch1},
\[L r_n^{\frac{1}{2}}\rho^{\frac{3}{2}}\leq \mu(\mathcal{C}(h_1))=\sum_{i=1}^2\mu(Q_i\cap\mathcal{C}(h_1)).\]
The symmetry of the set $\mathcal{C}(h_1)$ with respect to the axis $OY$ implies that the orthogonal transformation $\mathcal{O}:\mathbb{R}^2\longrightarrow\mathbb{R}^2$ such that $\mathcal{O}(x)=\mathcal{O}(x_1,x_2)=(-x_1,x_2)$ transforms $Q_1\cap\mathcal{C}(h_1)$ into $Q_2\cap\mathcal{C}(h_1)$ and then both sets measure the same, that is,
\[\mu(Q_1\cap\mathcal{C}(h_1))=\mu(Q_2\cap\mathcal{C}(h_1))=\frac{1}{2}\mu(\mathcal{C}(h_1)).\]
By Lemma \ref{Ch1contR2} we further have that, for $i=1,2$, $Q_i\cap\mathcal{C}(h_1)\subset \mathcal{C}(h_1)\subset B((\alpha-\rho)e_2,\alpha)$
and hence $\mu(Q_i\cap\mathcal{C}(h_1)\cap B((\alpha-\rho)e_2,\alpha))=\mu(Q_i\cap\mathcal{C}(h_1))\geq L^\mathcal{F} r_n^{\frac{1}{2}}\rho^{\frac{3}{2}},$
where $L^\mathcal{F}=\sqrt{2}/6$. This completes the proof of Proposition \ref{propositionmalos2}. \qed 

\end{proof}

Now, we define $\mathcal{U}_{0,r_n}=\mathcal{U}^{\mathcal{G}}_{0,r_n}\cup\mathcal{U}^{\mathcal{F}}_{0,r_n}$. As we mentioned at the beginning of Proposition \ref{cerca.R2}, $\mathcal{U}_{x,r_n}=\{T(U), \ U\in \mathcal{U}_{0,r_n}\}$
is a finite family with $m_2=m^\mathcal{G}+m^\mathcal{F}=6$ elements satisfying that, for each $U\in \mathcal{U}_{0,r_n}$,
\[P_X(T(U))\geq \delta\mu(T(U)\cap B(P_{\Gamma}x-\alpha\eta,\alpha))=\delta\mu(U\cap B((\alpha-\rho)e_2,\alpha))\geq L_2r_n^{\frac{1}{2}}\rho^{\frac{3}{2}},\]
where $L_2=\delta\min(L^\mathcal{G},L^\mathcal{F})$. This completes the proof of Proposition \ref{cerca.R2}. \qed 

\end{proof}

We are know in position to complete the proof of Theorem \ref{ordenR2}. Recall that, if we define for each $x\in S$ a family {\mbox{$\mathcal{U}_{x,r_n}$}} unavoidable
and finite for $\mathcal{E}_{x,r_n}$, then
$$
\mathbb{E}(d_\mu(S,S_n))\leq\int_S{\sum_{U\in\mathcal{U}_{x,r_n}}\exp(-nP_X(U))\mu(dx)}.
$$
We divide $S$ into two subsets $S=\left\{x\in S:\ \ d(x,\partial S)>\frac{r_n}{2}\right\}\cup\left\{x\in S:\ \ d(x,\partial S)\leq \frac{r_n}{2}\right\}$
and then
\begin{eqnarray}
\mathbb{E}(d_\mu(S,S_n))&\leq& \int_{\left\{x\in S:\ \ d(x,\partial S)>\frac{r_n}{2}\right\}}{\sum_{U\in\mathcal{U}_{x,r_n}}\exp(-nP_X(U))\mu(dx)}\nonumber\\
&+&\int_{\left\{x\in S:\ \ d(x,\partial S)\leq\frac{r_n}{2}\right\}}{\sum_{U\in\mathcal{U}_{x,r_n}}\exp(-nP_X(U))\mu(dx)}.\label{buf4d}
\end{eqnarray}
For those $x\in S$ such that $d(x,\partial S)> r_n/2$ we make use of the families $\mathcal{U}_{x,r_n}$ given in Proposition \ref{lejos.R2} which ensures
the existence of suitable finite families $\mathcal{U}_{x,r_n}$ and provides a lower bound on the probability of the sets $U$, independent of $x$. Thus,
\begin{gather}
\int_{\left\{x\in S:\ \ d(x,\partial S)>\frac{r_n}{2}\right\}}{\sum_{U\in\mathcal{U}_{x,r_n}}\exp(-nP_X(U))\mu(dx)}\nonumber\\
\leq\int_{\left\{x\in S:\ \ d(x,\partial S)>\frac{r_n}{2}\right\}}{m_1\exp(-nL_1r_n^2)\mu(dx)}
=O\left(\textnormal{e}^{-L_1nr_n^2}\right).\label{O1d}
\end{gather}
For those $x\in S$ such that $d(x,\partial S)\leq r_n/2$, we may consider the unavoidable families $\mathcal{U}_{x,r_n}$ given in Proposition \ref{cerca.R2}. We have that
\begin{gather*}
\int_{\left\{x\in S:\ \ d(x,\partial S)\leq\frac{r_n}{2}\right\}}{\sum_{U\in\mathcal{U}_{x,r_n}}\exp(-nP_X(U))\mu(dx)}\nonumber\\
\leq\int_{\left\{x\in S:\ \ d(x,\partial S)\leq\frac{r_n}{2}\right\}}{m_2\exp\left(-L_2nr_n^{\frac{1}{2}}d(x,\partial S)^{\frac{3}{2}}\right)\mu(dx)}\nonumber\\
=\int_{\mathcal{T}^{-1}([0,r_n/2])}{g(\mathcal{T}(x))\mu(dx)},\label{uyyd}
\end{gather*}
where $\mathcal{T}:S\rightarrow\mathbb{R}$ is defined as $\mathcal{T}(x)=d(x,\partial S)$ and $g(z)=m_2\exp(-L_2nr_n^{\frac{1}{2}}z^{\frac{3}{2}})$. It follows from the change of variables formula (see Theorem 16.12 of \cite{Billingsley}) that
\begin{equation}\label{bill.cvt}
\int_{\mathcal{T}^{-1}([0,r_n/2])}{g(\mathcal{T}(x))\mu(dx)}=\int_{[0,r_n/2]}{g(\rho)\mu\mathcal{T}^{-1}(d\rho)}
\end{equation}
where $\rho=\mathcal{T}(x)$ and $\mu \mathcal{T}^{-1}$ is the measure on $\mathbb{R}$ defined by
$\mu \mathcal{T}^{-1}(A)=\mu(\mathcal{T}^{-1}(A)),$
for $\ A\subset \mathbb{R}.$ The measure $\mu \mathcal{T}^{-1}$ is characterized by
$F(z)=\mu\{x\in S:\ d(x,\partial S)\leq z\}.$
Since ${\textnormal{reach}}(\partial S)\geq \alpha$,  $F(z)$ is a polynomial of degree at most $2$ in $z$, $0\leq z<\alpha$, see \cite{Federer59}. Therefore, it is a differentiable function and $F^\prime(z)$ is bounded on compact sets. In short, we obtain
\begin{eqnarray*}
\int_{[0,r_n/2]}{g(\rho)\mu\mathcal{T}^{-1}(d\rho)}
&=&\int_{[0,r_n/2]}{m_2\exp\left(-L_2nr_n^{\frac{1}{2}}\rho^{\frac{3}{2}}\right)F^\prime(\rho)d\rho}\\
&\leq& K \int_{0}^{\frac{r_n}{2}}{m_2\exp\left(-L_2nr_n^{\frac{1}{2}}\rho^{\frac{3}{2}}\right)d\rho}\\
&=&K\int_{0}^{\frac{L_2n}{2^{3/2}}r_n^2}{m_2\frac{1}{\frac{3}{2}L_2^{2/3}}}r_n^{-\frac{1}{3}}n^{-\frac{2}{3}}\textnormal{e}^{-v}v^{-\frac{1}{3}}dv
=O\left(r_n^{-\frac{1}{3}}n^{-\frac{2}{3}}\right),\label{O2d}
\end{eqnarray*}
where we have used the change of variables formula $v=L_2nr_n^{\frac{1}{2}}\rho^{\frac{3}{2}}$ and also the fact
that $\int_{0}^\infty{\textnormal{e}^{-v}v^{-\frac{1}{3}}dv}<\infty$. Turning to the
computation of $\mathbb{E}(d_\mu(S,S_n))$ in (\ref{buf4d}), it follows from (\ref{O1d}) and (\ref{O2d}) that
\begin{equation}\label{larga}
\mathbb{E}(d_\mu(S,S_n))=O\left(\textnormal{e}^{-L_1nr_n^2}+r_n^{-\frac{1}{3}}n^{-\frac{2}{3}}\right).
\end{equation}
Since $r_n$ is bounded by $\alpha$ and $n r_n^2/\log n$ goes to infinity, we have $\textnormal{e}^{-L_1nr_n^2}=o(r_n^{-1/3}n^{-2/3})$. Therefore,
$\mathbb{E}(d_\mu(S,S_n))=O(r_n^{-1/3}n^{-2/3})$,
which completes the proof of Theorem \ref{ordenR2}. \qed 
\end{proof}

\begin{proof}[Theorem \ref{nomejor}]
Let $S=B(0,\alpha)$ and assume that the distribution $P_X$ is uniform on $S$.  Our aim is to find a lower bound for $\mathbb{E}(d_\mu(S,S_n))$. Thus,
\begin{eqnarray*}
\mathbb{E}(d_\mu(S,S_n))&=&\int_S{P(\exists y\in {B}(x,r_n): {B}(y,r_n)\cap\mathcal{X}_n=\emptyset)\mu(dx)}\\
&\geq&\int_{\left\{x\in S:\ \ d(x,\partial S)\leq\frac{r_n}{2}\right\}}{P(\exists y\in {B}(x,r_n): {B}(y,r_n)\cap\mathcal{X}_n=\emptyset)\mu(dx)}.
\end{eqnarray*}
For each $x\in S$ such that $d(x,\partial S)\leq r_n/{2}$ let $\eta=x/\left\|x\right\|$ and $\tilde{x}=(\left\|x\right\|+r_n)\eta$, see Figure \ref{fig2:apuf}. A simple geometric argument shows
that $P_X(B(\tilde{x},r_n))\leq 1/2$. 
Since $\tilde{x}\in B(x,r_n)$ we have

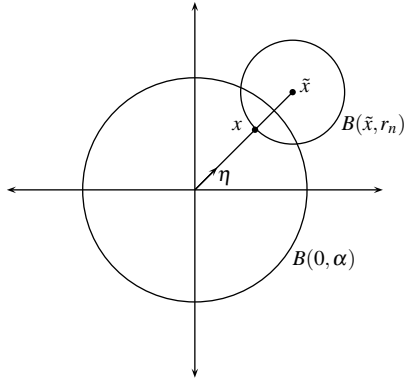
\begin{figure}[h]
\begin{center}
\setlength{\unitlength}{1mm}
          \begin{picture}(50,50)
              \psline[linewidth=.5pt]{<->}(0,2.5)(5,2.5)
         \psline[linewidth=.5pt]{<->}(2.5,0)(2.5,5)

\pscircle[linewidth=.5pt](2.5,2.5){1.5}
\pscircle[linewidth=.5pt](3.8,3.8){0.7}
\psdot[dotsize=2pt .5](3.3,3.3)
\psline[linewidth=.5pt]{}(2.5,2.5)(3.8,3.8)
\psdot[dotsize=2pt .5](3.8,3.8)
\psline[linewidth=.5pt]{->}(2.5,2.5)(2.8,2.8)
\put(30,33){\scriptsize{$x$}}
\put(39,38){\scriptsize{$\tilde{x}$}}
  \put(28,26.2){\scriptsize{$\eta$}}
    \put(38,15){\scriptsize{$B(0,\alpha)$}}
    \put(44.5,33){\scriptsize{$B(\tilde{x},r_n)$}}

           \end{picture}
\caption{\textit{Given $x\in B(0,\alpha)$ such that $d(x,\partial S)\leq r_n/2$, we define $\tilde{x}$.}}
\label{fig2:apuf}
\end{center}
\end{figure}

\begin{eqnarray}
\mathbb{E}(d_\mu(S,S_n))&\geq&\int_{\left\{x\in S:\ \ d(x,\partial S)\leq\frac{r_n}{2}\right\}}{(1-P_X(B(\tilde{x},r_n)))^n\mu(dx)}\nonumber\\
&\geq&\int_{\{x\in S: d(x,\partial S)\leq r_n/2\}}\exp\left(\frac{-nP_X(B(\tilde{x},r_n))}{1-P_X(B(\tilde{x},r_n))}\right)\mu(dx)\nonumber\\
&\geq&\int_{\{x\in S: d(x,\partial S)\leq r_n/2\}}\exp\left(-2nP_X(B(\tilde{x},r_n))\right)\mu(dx).\label{eq23}
\end{eqnarray}
Above we used the fact that $(1-z)^n\geq\exp(-nz/(1-z))$ for $z\in [0,1)$. 
In view of (\ref{eq23}) we need again an upper bound for $P_X(B(\tilde{x},r_n))$. The previous bound $P_X(B(\tilde{x},r_n))\leq 1/2$ is
too rough for our purposes and obviously it can be sharpened. Let us consider the composed
function formed by first applying an
orthogonal transformation $\mathcal{O}:\mathbb{R}^2\rightarrow\mathbb{R}^2$ such that $\mathcal{O}(\eta)=-e_2$ and then
applying the translation by the vector $(\alpha-d(x,\partial S))e_2$, see Figure \ref{fig2:apuf2}. Since the Lebesgue measure is invariant
under orthogonal transformations and translations we have that
$\mu(B(\tilde{x},r_n)\cap S)=\mu(B(-r_ne_2,r_n)\cap B((\alpha-d(x,\partial S))e_2,\alpha))$. The set $B(-r_ne_2,r_n)\cap B((\alpha-d(x,\partial S))e_2,\alpha)$ is
the intersection of two balls with radius $r_n$ and $\alpha$ such that the distance between their centres is equal to $\alpha+r_n-d(x,\partial S)$. Recall that
this set appeared in Proposition \ref{propositionmalos2}. Following the notation used previously, $B(-r_ne_2,r_n)\cap B((\alpha-d(x,\partial S))e_2,\alpha)=\mathcal{C}(h_1)\cup\mathcal{A}(h_2)$,
where $\mathcal{C}(h_1)$ is given by (\ref{DefCh1}) and
\[\mathcal{A}(h_2)=\{z\in\mathbb{R}^2: -(h_1+h_2)\leq \left\langle z,e_2\right\rangle\leq -h_1\}\cap B((\alpha-d(x,\partial S))e_2,\alpha).\]
Recall that the values of $h_1$ and $h_2$ were easily deduced from the Pythagorean theorem by solving the system
\begin{equation*}
\left\{\begin{array}{l}
(r_n-h_1)^2+\lambda^2=r_n^2,\\
(\alpha-h_2)^2+\lambda^2=\alpha^2,\\
h_1+h_2=d(x,\partial S).
\end{array}\right.
\end{equation*}
Thus,
\[h_1=\frac{d(x,\partial S)(2\alpha-d(x,\partial S))}{2(\alpha+r_n-d(x,\partial S))},\ \ \ h_2=\frac{d(x,\partial S)(2r_n-d(x,\partial S))}{2(\alpha+r_n-d(x,\partial S))}.\]
\begin{figure}[h]
\begin{center}
\setlength{\unitlength}{1mm}
          \begin{picture}(135,60)
              \psline[linewidth=.5pt]{<->}(0.5,2.5)(4.5,2.5)
         \psline[linewidth=.5pt]{<->}(2.5,0)(2.5,5)

\pscircle[linewidth=.5pt](2.5,2.5){1.5}
\pscircle[linewidth=.5pt](3.8,3.8){0.7}
\psdot[dotsize=2pt .5](3.3,3.3)
\psline[linewidth=.5pt]{}(2.5,2.5)(3.8,3.8)
\psdot[dotsize=2pt .5](3.8,3.8)
\psline[linewidth=.5pt]{->}(2.5,2.5)(2.8,2.8)
\put(30,33){\scriptsize{$x$}}
\put(39,38){\scriptsize{$\tilde{x}$}}
  \put(28,26.2){\scriptsize{$\eta$}}
    \put(38,15){\scriptsize{$B(0,\alpha)$}}
    \put(28,47){\scriptsize{$B(\tilde{x},r_n)$}}

        \psline[linewidth=.5pt]{<->}(5,2.5)(9,2.5)
         \psline[linewidth=.5pt]{<->}(7,0)(7,5)

\pscircle[linewidth=.5pt](7,2.5){1.5}
\pscircle[linewidth=.5pt](7,0.66){0.66}
\psdot[dotsize=2pt .5](7,0.66)
\put(70.3,6.8){\tiny{$\mathcal{O}(\tilde{x})$}}

\psline[linewidth=.5pt]{->}(7,2.5)(7,2.1)
  \put(71,23){\scriptsize{$-e_2$}}
    \put(83,15){\scriptsize{$B(0,\alpha)$}}

\pscustom[linewidth=.05pt,fillstyle=solid, fillcolor=black]{
\psarc(11.5,3.68){1.5}{260}{280}
\psline{-}(12.,2.27)(11,2.27)
}
\pscustom[linewidth=.05pt,fillstyle=solid, fillcolor=gray]{
\psline{-}(11,2.27)(12.,2.27)
\psarc(11.5,1.84){0.66}{50}{130}
}
 \psline[linewidth=.5pt]{|-|}(13,2.27)(13,2.5)
 \psline[linewidth=.5pt]{-|}(13,2.27)(13,2.17)
 \put(132,21.5){\tiny{$h_2$}}
 \put(132,23.2){\tiny{$h_1$}}
 \psline[linestyle=dashed,linewidth=.25pt]{-}(11.5,2.27)(13,2.27)
 \psline[linestyle=dashed,linewidth=.25pt]{-}(11.5,2.18)(13,2.18)
       \psline[linewidth=.5pt]{<->}(9.5,2.5)(13.5,2.5)
         \psline[linewidth=.5pt]{<->}(11.5,0)(11.5,5)

\pscircle[linewidth=.5pt](11.5,3.68){1.5}
\pscircle[linewidth=.5pt](11.5,1.84){0.66}
\psdot[dotsize=2pt .5](11.5,1.84)
\psdot[dotsize=2pt .5](11.5,3.68)
    \put(117,10){\scriptsize{$B(-r_ne_2,r_n)$}}
\put(47,48){{{{$\longrightarrow$}}}}
\put(49,50){{{\scriptsize{$\mathcal{O}$}}}}

\put(88,48){{{{$\longrightarrow$}}}}
\put(77,50){{{\scriptsize{$\oplus\{(\alpha-d(x,\partial S))e_2\}$}}}}

 \put(24,55){\scriptsize{(a)}}
  \put(69,55){\scriptsize{(b)}}
   \put(114,55){\scriptsize{(c)}}
           \end{picture}
\caption{\textit{(a) $B(\tilde{x},r_n)\cap S$. (b) Result of applying an orthogonal transformation $\mathcal{O}:\mathbb{R}^2\rightarrow\mathbb{R}^2$ such that $\mathcal{O}(\eta)=-e_2$. (c) Translation by the vector $(\alpha-d(x,\partial S))e_2$. In black $\mathcal{A}(h_2)$ and in gray $\mathcal{C}(h_1)$.}}
\label{fig2:apuf2}
\end{center}
\end{figure}
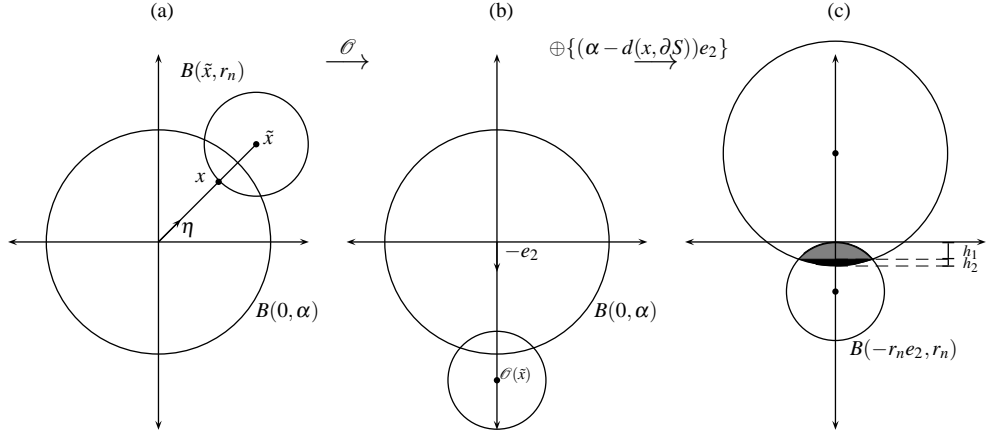

%
%
%
%
%
%
%
Since $\mathcal{C}(h_1)$ and $\mathcal{A}(h_2)$ are disjoint, up to a zero measure set, we have
\begin{equation}\label{AyC}
\mu(B(-r_ne_2,r_n)\cap B((\alpha-d(x,\partial S))e_2,\alpha))=\mu(\mathcal{C}(h_1))+\mu(\mathcal{A}(h_2)).
\end{equation}
First, in order to find an upper bound in (\ref{AyC}), we shall see that $\mu(\mathcal{A}(h_2))\leq\mu(\mathcal{C}(h_1))$. It can be easily
proved that $\mu(\mathcal{A}(h_2))=\mu(\mathcal{A}_0(h_2))$, where
\[\mathcal{A}_0(h_2)=\{z\in\mathbb{R}^2:\ 0\leq\left\langle z,e_2\right\rangle \leq h_2\}\cap B(-(\alpha-h_2)e_2,\alpha).\]
As in Lemma \ref{Ch1}, we have $\mu(\mathcal{A}_0(h_2))=\int_0^{h_2}2\sqrt{2\alpha y-y^2}dy$.  Using the change of variable $l=h_2-y$, and taking into account that
$2\alpha y-y^2=\alpha^2-(\alpha-y)^2$  we can write
$$\mu(\mathcal{A}_0(h_2))=\int_0^{h_2}2\sqrt{\alpha^2-(\alpha-h_2+l)^2}dl=2\int_0^{h_2}\sqrt{s(l)}dl,$$
Similarly we have $\mu(\mathcal{C}(h_1))=2\int_0^{h_1}\sqrt{r(l)}dl,$ where $r(l)=r_n^2-(r_n-h_1+l)^2$. Note that $r(0)=s(0)=\lambda^2$ and $h_2\leq h_1$.
It is easy to show that $s(l)\leq r(l)$ and therefore
\[\mu(\mathcal{A}(h_2))=2\int_{0}^{h_2}\sqrt{s(l)}dl\leq 2\int_{0}^{h_2}\sqrt{r(l)}dl\leq\mu(\mathcal{C}(h_1)).\]
Now, if we return to Equation (\ref{AyC}), we get
\begin{equation*}
\mu(B(\tilde{x},r_n)\cap S)\leq 2\mu(\mathcal{C}(h_1)).
\end{equation*}
An upper bound for $\mu(\mathcal{C}(h_1))$ can be easily found since 
$$
\mu(\mathcal{C}(h_1))
=\int_{0}^{h_1}\sqrt{\left(2r_ny-y^2\right)}dy\leq \int_{0}^{h_1}\sqrt{2 r_n y}dy=O\left(r_n^{\frac{1}{2}}h_1^{\frac{3}{2}}\right)=O\left(r_n^{\frac{1}{2}}d(x,\partial S)^{\frac{3}{2}}\right),
$$
where in the last equality we have used $h_1\leq d(x,\partial S)$. As a consequence,
\[P_X(B(\tilde{x},r_n))\leq Lr_n^{\frac{1}{2}}d(x,\partial S)^{\frac{3}{2}},\]
with $L>0$ a constant which does not depend on $x$.
Finally, if we apply the latter bound to (\ref{eq23}), then we have that
\begin{eqnarray*}
\mathbb{E}(d_\mu(S,S_n))&\geq&\int_{\{x\in S: d(x,\partial S)\leq r_n/2\}}\exp\left(-2nLr_n^{\frac{1}{2}}d(x,\partial S)^{\frac{3}{2}}\right)\mu(dx)\\
&=&\int_{\mathcal{T}^{-1}([0,r_n/2])}{g(\mathcal{T}(x))\mu(dx)},\\
\end{eqnarray*}
where $\mathcal{T}:S\rightarrow\mathbb{R}$ is defined as $\mathcal{T}(x)=d(x,\partial S)$ and $g(z)=\exp(-2nLr_n^{\frac{1}{2}}z^{\frac{3}{2}})$. We use the same
change of variables formula, see (\ref{bill.cvt}), with  $F(z)=\pi(\alpha^2-(\alpha-z)^2)$. So
\begin{eqnarray*}
\mathbb{E}(d_\mu(S,S_n))&\geq&\int_{0}^{r_n/2}\exp\left(-2nLr_n^{\frac{1}{2}}\rho^{\frac{3}{2}}\right)F^\prime{}(\rho)d\rho\\
=\int_{0}^{r_n/2}\exp\left(-2nLr_n^{\frac{1}{2}}\rho^{\frac{3}{2}}\right)2\pi (\alpha-\rho)d\rho&\geq& \pi\alpha\int_{0}^{r_n/2}\exp\left(-2nLr_n^{\frac{1}{2}}\rho^{\frac{3}{2}}\right)d\rho.
\end{eqnarray*}
A straightforward calculation shows that
$$
\mathbb{E}(d_\mu(S,S_n))\geq Cr_n^{-\frac{1}{3}}n^{-\frac{2}{3}}\int_{0}^{\frac{Lnr_n^2}{\sqrt{2}}}e^{-v}v^{-\frac{1}{3}}dv   
$$
for some constant $C>0.$
Since $nr_n^2\rightarrow\infty$, we have
$$\liminf_{n\rightarrow\infty}r_n^{\frac{1}{3}}n^{\frac{2}{3}}\mathbb{E}(d_\mu(S,S_n))\geq C\int_{0}^{\infty}{e^{-v}v^{-\frac{1}{3}}dv}>0.$$
This completes the proof of Theorem \ref{nomejor}. \qed 
\end{proof}

\begin{proof}[Theorem \ref{extremos}]
Note that under the stated assumptions on $P_X$ we have that, for any measurable nonnegative function in $S$, $\varphi$, we have that
$$
\int_S \varphi(x) P_X(dx)\leq \beta \int_S \varphi(x)\mu(dx).
$$
Using this fact in (\ref{extremes}), we can follow the same lines as in the final
part of the proof of Theorem \ref{ordenR2} to easily conclude the result. \qed 

\end{proof}

\section*{Acknowledgements} The second author thanks Prof. Luc Devroye for his valuable help.
This work has been partially supported by the grant
MTM2008-03010 (A. Rodr\'{i}guez-Casal) from the Spanish Ministerio de
Educaci\'{o}n y Ciencia.

\bibliographystyle{plain}


\end{document}